%% file: main.tex
\documentclass[11pt,a4paper]{article}

\usepackage{microtype}
\usepackage{setspace}
\usepackage{geometry}
\usepackage{amsmath,amssymb,amsfonts,amsthm,mathtools}
\usepackage{derivative}
\usepackage{dsfont}

\newtheorem{thm}{Theorem}
\newtheorem{lem}{Lemma}

\newtheorem{rem}{Remark}
\newtheorem{coro}{Corollary}

\newtheorem{assum}{Assumption}

\theoremstyle{definition}
\newtheorem{exa}[thm]{Example}
\renewcommand{\qedsymbol}{$\square$}

\usepackage{graphicx}
\usepackage{tikz}
\usetikzlibrary{decorations.pathreplacing,calligraphy}
\usepackage{tabularx,tabulary,array,booktabs,multirow}
\usepackage{etoolbox}
\usepackage{xparse}
\makeatletter\g@addto@macro\@floatboxreset\centering\makeatother 

\BeforeBeginEnvironment{tabular*}{\footnotesize}

\NewDocumentCommand{\note}{o g}{%
\parbox{\textwidth}{\footnotesize\vspace*{10pt}%
\IfValueT{#1}{\textit{#1}:\quad}%
#2}}
\usepackage{caption}
\usepackage{subcaption}
\usepackage{placeins}

\usepackage[all]{xy}

\usepackage{enumerate}
\usepackage{verbatim}
\usepackage{fancyhdr}

\usepackage[hidelinks]{hyperref}
\usepackage[toc,page]{appendix}

\usepackage[authordate, ibidtracker=false]{biblatex-chicago}
\tikzset{
  arrow/.style      = {thick,->,>=stealth},
  arrow2/.style     = {very thick,-,>=stealth,draw=blue},
  arrow_dash/.style = {thick,dashed,-,>=stealth,draw=blue},
  arrow_main/.style = {very thick,->,>=stealth,draw=green},
  arrowlog/.style   = {very thick,->,>=stealth,draw=blue},
  stddash/.style    = {thick,dashed,->,>=stealth,draw=black}
}

\newcommand{\norm}[1]{\left\lVert#1\right\rVert}

\newcommand{\E}{\mathbb{E}}
\newcommand{\var}{\operatorname{var}\!}
\newcommand{\cov}{\operatorname{cov}\!}

\providecommand{\keywords}[1]{%
  {\small\textit{Keywords:} #1}%
}
\providecommand{\jel}[1]{%
  {\small\textit{JEL codes:} #1}%
}

\usepackage{titlesec}
\titleformat*{\section}{\centering\large\bfseries}
\titleformat*{\subsection}{\bfseries}
\titleformat{\paragraph}[runin]{\itshape}{}{}{}[.]
\titlelabel{\thetitle.\quad}

\allowdisplaybreaks

\title{Generalized AKM: Flexible Controls and Interactions in Wage Decompositions\footnote{%
We are grateful to Edoardo M. Acabbi, Stanislav Anatolyev, Andrea Alati, Matias D. Cattaneo, Patrick Kline, Michal Kolesár, Mikkel Plagborg-Møller, Ulrich Müller, and seminar participants at Princeton University, XXIII Brucchi Luchino conference, and IRLE Berkeley for fruitful comments and suggestions. 
This paper results from research funded under the umbrella of the ERC-CZ project No. LL2319. Funding from the European Union’s Horizon 2020 research and innovation program under the Marie Skłodowska-Curie grant agreement No. 870245 is also gratefully acknowledged.
\vspace{.05cm}
}}
\author{Francesco Del Prato%
\footnote{\emph{Department of Economics and Business Economics, Aarhus University}. \emph{e}-mail: francesco.delprato@econ.au.dk.} $\quad$ 
Yaroslav Korobka%
\footnote{\emph{CERGE-EI (Charles University \& Czech Academy of Sciences)}. \emph{e}-mail: yaroslav.korobka@cerge-ei.cz.} $\quad$
Paolo Zacchia%
\footnote{\emph{Department of Economics, Ca' Foscari University of Venice}; \emph{CERGE-EI}; \emph{CEPR}; and \emph{IZA}. \emph{e}-mail: Paolo.Zacchia@cerge-ei.cz.}
}
\date{July 2026}

\begin{document}

\maketitle
\thispagestyle{empty}

\begin{abstract}
\noindent
    How much wage dispersion is attributed to workers, firms, and their sorting depends on how wages are adjusted for observed characteristics. Standard AKM decompositions impose a known linear adjustment. 
    We develop Generalized AKM, a framework that permits an unknown smooth covariate function and group-specific nonlinear interactions while preserving the original variance components. 
    We prove consistency and asymptotic normality with heteroskedastic errors and many fixed effects, and characterize the stronger smoothness required for quadratic forms.
    In Portuguese employer--employee data, adding worker and firm-input controls lowers the bias-corrected worker-effect variance from $0.551$ to $0.474$ of total wage variance, firm-effect variance from $0.144$ to $0.121$, and sorting from $0.080$ to $0.047$.
    Across three group-specific nonlinear bases, firm-effect variance remains between $0.114$ and $0.117$ and sorting between $0.041$ and $0.042$, while worker-effect variance ranges from $0.474$ to $0.491$.
    Which controls enter matters more for firm variance and sorting than how flexibly they enter; worker variance remains more sensitive to the basis.
\end{abstract}

\par
\jel{J31, J62, C14, C23, C55}

\newpage

\section{Introduction}
\setcounter{page}{1}

How much wage dispersion is attributed to workers, firms, and the sorting of workers across firms?
AKM decompositions answer this question by assigning wage variation to worker effects, firm effects, and their covariance after adjusting wages for observed characteristics \parencite{abowd99}.
These components have become standard measures of worker heterogeneity, firm pay differences, and sorting in matched employer--employee data.
So how well wages are adjusted for observables, before the decomposition is even run, shapes how these components should be read.
Standard AKM estimators, including the leave-out correction of \textcite{kline}, impose a known linear control function.
Yet returns to experience and schooling are nonlinear \parencite{mincer, card18}, and firm inputs need not be additively separable from worker characteristics.
If the control function is misspecified, its omitted nonlinearities and interactions can be reassigned to worker effects, firm effects, or sorting.
At stake is whether the decomposition isolates persistent worker and firm heterogeneity or partly relabels observed wage schedules as fixed effects.

We develop Generalized AKM, a semiparametric estimator for the original AKM variance components when observed covariates enter through an unknown smooth function.
The estimator approximates this function with a growing polynomial or spline basis and permits the resulting nonlinear profiles to vary across worker groups.
Worker and firm fixed effects, year effects, and all group-specific series terms enter one augmented regressor set, rather than adjusting wages for covariates in a first step and decomposing the residual.\footnote{The two-step alternative introduces a higher-order bias, because the first-step estimation error enters both the adjusted outcome and the residual used in the bias correction \parencite{kline2024}.}
The construction preserves the familiar worker, firm, and sorting components rather than replacing them with latent types or a different earnings model.
It answers two related questions: how much does this added flexibility change the decomposition in practice, and when can these components be estimated at all without treating the covariate function as known?

The empirical application answers the first question.
We use Portuguese linked employer--employee data for 2008--2018 and estimate a five-step ladder of increasingly flexible specifications. The common estimation sample is a leave-match-out set (the largest connected set of workers and firms that supports the match-level correction) of 9.4 million person-year observations, 1.7 million workers, and 105,000 firms.
Adding the observed controls---quadratic and cubic age terms, education, qualification, and firm inputs---in a linear additive specification lowers worker-effect variance from $0.551$ to $0.474$ of total wage variance, firm-effect variance from $0.144$ to $0.121$, and sorting from $0.080$ to $0.047$.\footnote{These are bias-corrected point estimates. The correction used in the application treats each worker--firm match as a cluster, which places it outside the observation-level inference theory of Section \ref{sec:lim}; we therefore report no standard errors and read differences across specifications as descriptive.}
The proportional movement is largest for sorting, which falls by more than two fifths: a substantial part of what the uncontrolled decomposition attributes to high-wage workers matching with high-wage firms is instead accounted for by education, qualification, and firm inputs.
Allowing a common degree-5 polynomial and then interacting the basis with gender, education, and qualification changes firm variance and sorting much less.
Across polynomial, Hermite, and cubic B-spline implementations of the group-specific nonlinear model, firm variance lies between $0.114$ and $0.117$ and sorting between $0.041$ and $0.042$; worker variance ranges from $0.474$ to $0.491$.
Thus which observables enter matters more for measured firm variance and sorting than how flexibly they are allowed to enter once they are included.
The worker component is more sensitive to the nonlinear basis.

The limited movement across flexible specifications is itself economically informative.
Predictive content and decomposition movement are distinct objects.
A richer wage adjustment changes the AKM decomposition only when its fitted component reallocates wage variation along the worker--firm network: it must change the second moments of the worker and firm effects or their covariance.
Nonlinear terms can improve conditional wage fit while leaving those moments nearly unchanged.
On a held-out sample drawn so that the training data keep the worker--firm network connected, the fully interacted specification lowers mean squared error by $1.1\%$ relative to heterogeneous linear controls, yet firm variance and sorting remain nearly unchanged.
Our specification ladder separates two margins that standard robustness exercises often combine.
The large movement from adding education, qualification, and firm inputs shows that observed worker composition and firm conditions matter.
The smaller movement from letting those controls enter more flexibly indicates that measured firm heterogeneity and sorting are not mainly consequences of imposing additive linear controls; the greater basis sensitivity of worker variance identifies where the choice of approximating basis still matters.
Generalized AKM turns this interpretation into an estimable restriction by asking whether the conventional components survive a much broader class of wage adjustments.

The second question is theoretical: when can these components be estimated at all without treating the covariate function as known? The problem is not a routine extension of either leave-out estimation or semiparametric series regression.
Three sources of error must be handled jointly: many-regressor bias from high-dimensional fixed effects, heteroskedasticity, and approximation error from the unknown control function.
Existing leave-out estimators address the first two under linear controls, while standard series results study approximation error for linear functionals of the coefficients, such as a single slope.
A variance component is not of that kind: it is quadratic in the worker and firm effects, being a sum of squares and cross-products rather than a weighted average.
That difference is what makes the problem hard: approximation error then enters both the plug-in quadratic form and its leave-out bias correction.
We derive conditions under which both terms vanish together and show that the estimator is consistent and asymptotically normal for the original variance components.
The key smoothness requirement is stronger than for linear functionals: the unknown function must have more continuous derivatives than the number of continuous covariates entering flexibly.
The condition is nonetheless mild in wage applications. With age as the only continuous covariate, it asks that the wage--age profile be twice continuously differentiable; with the three firm inputs added, that the wage surface in age and inputs have five derivatives.
We also establish a chi-squared limit when the target has fixed rank, a normal limit when its rank grows with sample size, and a random-projection approximation for large administrative datasets.

The simulations show what this added flexibility buys.
When the control function is strongly nonlinear, higher-degree bases reduce the bias, especially in high-leverage designs.
When the relationship is linear or only mildly nonlinear, the flexible estimator performs similarly to the linear alternatives, so extra basis terms impose little cost.
In a calibration that holds the Portuguese worker--firm graph fixed, the leave-out correction has point-estimate bias below $0.2\%$ in all four nuisance designs, compared with $1.2$--$1.6\%$ for the homoskedastic correction and $1.9$--$2.4\%$ for the plug-in estimator.
Finite-sample confidence intervals remain less reliable in several high-leverage stress designs, so the strongest numerical evidence concerns the decomposition itself.

The central contribution is to make AKM variance components estimable with a flexible and interacted covariate function without changing the objects being estimated.
Relative to \textcite{kline}, the estimator allows the approximation space to grow and states when leave-out validity survives the resulting approximation error.
Relative to the semiparametric series literature \parencite{donald, cattaneoalt, cattaneo18}, the target is a quadratic form rather than a slope coefficient, which generates the stronger smoothness condition.
The empirical application connects this result to work on the sensitivity of AKM decompositions \parencite{andrews, bonhomme, bonhomme19}: it separates the effect of adding observed controls from the effect of allowing nonlinearities and interactions among them.
Unlike grouped-firm and latent-type approaches \parencite{bonhomme19}, Generalized AKM retains the original worker and firm variance components and asks how their estimated magnitudes change when the control function is no longer imposed.

\paragraph{Roadmap}
Section \ref{sec:setup} defines the semiparametric decomposition and explains why existing estimators do not directly apply.
Sections \ref{sec:finite}--\ref{sec:lim} develop the leave-out estimator, its large-sample properties, and its random-projection approximation.
The simulation in Section \ref{sec:sim} isolates the gains from flexibility, and the Portuguese application in Section \ref{sec:empirical} shows how adding observables, nonlinearities, and interactions changes the decomposition.

\section{Setup}\label{sec:setup}
\subsection{Model and parameter of interest}
We consider the following semiparametric model,
\begin{equation}\label{main}
    y_i = x_i'\beta + f(z_i) + e_i, \quad i = 1, \ldots, n,
\end{equation}
where the regressors $x_i \in \mathbb{R}^p$, $z_i \in \mathbb{R}^d$ are non-random. The unknown function $f(z)$ belongs to the class of smooth and real-valued functions, $f \in \mathcal{F}$. The unobserved errors $\{e_i\}_{i = 1} ^n$ are mutually independent and obey $\E[e_i] = 0$, but may possess observation-specific variances $\E[e_i^2] = \sigma_i ^2$. In the wage application, $x_i$ collects the high-dimensional worker and firm indicators whose dispersion we ultimately want to measure, while $z_i$ collects the observed covariates (age, firm inputs, and so on) whose effect on wages we want to control for without committing to a functional form.

Our object of interest is a quadratic form $\theta := \beta' A \beta$ for some known non-random symmetric matrix $A \in \mathbb{R}^{p \times p}$ of rank $r$. The matrix $A$ selects which feature of $\beta$ to isolate: different choices pick out the variance of worker effects, the variance of firm effects, or their covariance.

\subsection{Examples}
The quadratic form covers many quantities of economic interest. The two examples below are the analysis-of-variance and two-way fixed-effects cases studied by \textcite{kline}, whose notation and framing we follow so that the estimands remain directly comparable; the difference in each case is that the covariates enter through the unknown $f$ rather than linearly.

\begin{exa}[Generalized analysis of variance]\label{ex1}
    Let the data consist of $N$ groups, with group $g$ contributing $T_g$ observations. Since \textcite{fisher}, the standard analysis of variance model assumes:
    \begin{equation*}
        y_{gt} = \alpha_g + z_{gt}' \delta + \varepsilon_{gt}, \quad g = 1, \ldots, N, \quad t = 1, \ldots, T_g,
    \end{equation*}
    where $\alpha_g$ are group-specific fixed effects, and $z_{gt}$ is a vector of exogenous covariates. The focus here is the variability in the outcome variable attributable to groups, or
    \begin{equation*}
        \sigma_\alpha ^2 = \frac{1}{n} \sum_{g = 1} ^N T_g \left(\alpha_g - \bar{\alpha}\right)^2
    \end{equation*}
    with $n := \sum_{g = 1} ^N T_g$, and $\bar{\alpha} := n^{-1} \sum_{g = 1} ^N T_g \alpha_g$. Our approach relaxes the linearity assumption on covariates and instead assumes any form of unspecified functional dependence, that is
    \begin{equation*}
        y_{gt} = \alpha_g + f(z_{gt}) + \varepsilon_{gt}, \quad f \in \mathcal{F}, \quad g = 1, \ldots, N, \quad t = 1, \ldots, T_g.
    \end{equation*}
    We can represent it as in (\ref{main}), defining $i := i(g, t)$ with $i(\cdot, \cdot)$ being a bijective function, $y_i := y_{gt}$, $z_i := z_{gt}$ and $e_i := \varepsilon_{gt}$,
    \begin{equation*}
        x_i := d_i, \quad \beta := (\alpha_1, \ldots, \alpha_N)', \quad d_i := (\mathds{1}\{g = 1\}, \ldots, \mathds{1}\{g = N\})'.
    \end{equation*}
    Our object of interest $\sigma_\alpha^2$ can be represented here as $\beta' A \beta$, with
    \begin{equation*}
        A := \begin{pmatrix}
            A_d' A_d & 0 \\
            0 & 0
            \end{pmatrix}, \quad A_d := \frac{1}{\sqrt{n}} (d_1 - \bar{d}, \ldots, d_n - \bar{d}), \quad \bar{d} := \frac{1}{n} \sum_{i = 1} ^n d_i.
    \end{equation*}
\end{exa}

\begin{exa}[Generalized AKM]\label{ex2}
    Our second and leading example is a classic wage decomposition model proposed in \textcite{abowd99}. It models the log wage determination as an additive function of worker fixed effects, firm fixed effects, and a linear function of strictly exogenous covariates. Specifically,
    \begin{equation*}
        y_{gt} = \alpha_g + \psi_{j(g, t)} + z_{gt}' \delta + \varepsilon_{gt}, \quad g = 1, \ldots, N, \quad t = 1, \ldots, T_g.
    \end{equation*}
    Here, $\alpha_g$ and $\psi_{j(g, t)}$ capture the $g$th worker and $j$th firm unobserved heterogeneity component respectively, and $z_{gt}$ is a vector of exogenous regressors. Each of the $n = \sum_{g = 1} ^N T_g$ person-year observations is employed at one of $J + 1$ firms, and the employer of worker $g$ in period $t$ is recorded by $j(\cdot, \cdot): \{1, \ldots, N\} \times \{1, \ldots, \max_g T_g\} \rightarrow \{0, \ldots, J\}$. One of the model's objectives is to quantify how much of the variability in log wages is determined by firms,
    \begin{equation*}
        \sigma_\psi ^2 = \frac{1}{n} \sum_{g = 1} ^N \sum_{t = 1} ^{T_g} \left(\psi_{j(g, t)} - \bar{\psi}\right)^2,
    \end{equation*}
    where $\bar{\psi} = \frac{1}{n} \sum_{g = 1} ^N \sum_{t = 1} ^{T_g} \psi_{j(g, t)}$. Given evidence that returns to worker characteristics are nonlinear \parencite{mincer, card18}, part of the measured variance of worker and firm fixed effects may stem from the linear restriction on $z_{gt}$. Our methodology replaces it with unspecified functional dependence,
    \begin{equation*}
        y_{gt} = \alpha_g + \psi_{j(g, t)} + f(z_{gt}) + \varepsilon_{gt}, \quad f \in \mathcal{F}, \quad g = 1, \ldots, N, \quad t = 1, \ldots, T_g.
    \end{equation*}
    This formulation can capture nonlinear interactions between worker-level and firm-level characteristics that vary across observed groups in the population.
    
    Given that the common covariates $z_{gt}$ and the firm assignments $j(\cdot, \cdot)$ obey a strict exogeneity condition, we can rewrite the equation above as in (\ref{main}) with
    \begin{equation*}
        x_i := (d_i', h_i')', \quad \beta := (\alpha', \psi')', \quad \alpha := (\alpha_1, \ldots, \alpha_N)' + \mathds{1}_N' \psi_0, \quad \psi = (\psi_1, \ldots, \psi_J)' - \mathds{1}_J' \psi_0,
    \end{equation*}
    defining $y_i$, $z_i$, and $e_i$ as in Example \ref{ex1}, and $h_i := (\mathds{1}\{j(g, t) = 1\}, \ldots, \mathds{1}\{j(g, t) = J\})'$. The parameter of interest $\sigma_\psi ^2$ then can be rewritten as $\beta' A_\psi \beta$ with
    \begin{equation*}
        A_\psi := \begin{pmatrix}
            0 & 0 & 0 \\
            0 & A_h' A_h & 0 \\
            0 & 0 & 0
        \end{pmatrix}, \quad A_h := \frac{1}{\sqrt{n}} (h_1 - \bar{h}, \ldots, h_n - \bar{h}), \quad \bar{h} := \frac{1}{n} \sum_{i = 1} ^n h_i.
    \end{equation*}
\end{exa}

\subsection{Why existing estimators do not directly apply}
Throughout, our object of interest is the quadratic form $\theta = \beta' A \beta$, defined on the high-dimensional fixed-effect coefficients, while the control function $f(z)$ is estimated through a growing approximation space. 
Several existing estimation strategies address parts of this setting; we review each in turn and explain why none covers the full problem.

\paragraph{KSS-style leave-out estimation}
\textcite{kline}, hereafter KSS, provide valid leave-out estimation of the same quadratic form $\theta = \beta' A \beta$ in finite-dimensional linear many-regressor models under arbitrary heteroskedasticity. 
Our setting preserves that object, but drops the assumption that observables enter through a finite-dimensional linear control. 
Once $f(z)$ is replaced by a series approximation $p_k(z)' \alpha$ with $k \rightarrow \infty$, validity is no longer automatic: both $\hat{\beta}$ and the leave-out variance estimators inherit approximation error, so one must re-establish that leave-out validity survives series approximation.\footnote{%
A separate strand of the literature relaxes the economic structure of AKM itself. 
For example, \textcite{bonhomme19} model latent worker and firm types and reduce dimensionality through grouping. 
These approaches replace the original variance component with a different earnings structure, rather than preserving $\theta = \beta' A \beta$ under an unknown nuisance function.}

\paragraph{Semiparametric series and orthogonalization}
Series estimators in partially linear models \parencite{donald, cattaneoalt} and orthogonalized or cross-fit procedures \parencite{chernozhukov2018double, bonhomme2025} allow an unknown nuisance function, but they estimate low-dimensional slope coefficients or other approximately linear functionals.
None provides inference for quadratic forms under the high-dimensional fixed-effect structure and heteroskedasticity that define our setting.

\paragraph{Quadratic-functional estimation}
\textcite{breunigchen22} study a closely related object: a leave-one-out sieve estimator of a quadratic functional, with minimax and adaptive rates, in a nonparametric instrumental variables model. Their setting has no high-dimensional fixed effects, and their smoothness threshold reflects the degree of ill-posedness of the inverse problem rather than many-regressor bias. Our contribution is the analogous threshold for a quadratic form in a many-fixed-effects leave-out design.

\section{Finite-sample properties}\label{sec:finite}
From now on, we restrict the analysis to model (\ref{main}). To derive an estimator of $\beta$, one must regress $y_i$ on $x_i$ and \emph{functions} of $z_i$. 
To this end, let $p^1(z), \ldots, p^k (z)$ be some \emph{approximating} functions, and let $p_k (z) := (p^1(z), \ldots, p^k(z)) \in \mathbb{R}^k$ collect them into a vector. 
We assume that the class of functions $\mathcal{F}$ to which $f$ belongs can be well approximated by linear combinations of the basis functions $p^1(z), \ldots, p^k(z)$.

Define $M_{ij}$ as a $(i, j)$th element of $M := I_n - P_k (P_k' P_k)^{-1} P_k'$ with
\begin{equation*}
    P_k := (p_k (z_1), \ldots, p_k (z_n)) \in \mathbb{R}^{n \times k},
\end{equation*}
and let the partialed-out design matrix $S_{xx} := \sum_{i = 1} ^n \sum_{j = 1} ^n M_{ij} x_i x_j'$ have a full column rank. Then, the estimator of $\beta$ is defined as
\begin{equation*}
    \hat{\beta} := S_{xx}^{-1} \sum_{i = 1} ^n \sum_{j = 1} ^n M_{ij} x_i y_j.
\end{equation*}

\noindent
Under some regularity conditions, the bias of $\hat{\beta}$ is negligible for large $k$. Rewrite it as:
\begin{equation}\label{decomp}
    \begin{split}
        \hat{\beta} &= S_{xx}^{-1}  \sum_{i = 1} ^n \sum_{j = 1} ^n M_{ij} x_i \left(x'_j \beta + f(z_j) + e_j\right) \\
        &= S_{xx}^{-1} \sum_{i = 1} ^n \sum_{j = 1} ^n M_{ij} x_i x_j' \beta + S_{xx}^{-1} \sum_{i = 1} ^n \sum_{j = 1} ^n M_{ij} x_i f(z_j) + S_{xx}^{-1} \sum_{i = 1} ^n \sum_{j = 1} ^n M_{ij} x_i e_j \\
        &= \beta + S_{xx}^{-1} \sum_{i = 1} ^n \sum_{j = 1} ^n M_{ij} x_i f(z_j) + S_{xx}^{-1} \sum_{i = 1} ^n \sum_{j = 1} ^n M_{ij} x_i e_j := \beta + \mathcal{B} + \mathcal{U}.
    \end{split}
\end{equation}

\noindent
Because of the zero-mean assumption on errors, $\E[\mathcal{U}] = 0$ holds, and the bias in the estimator is reflected in the $\mathcal{B}$ term only. Under some conditions discussed below,
\begin{equation*}
    \mathcal{B} = o(1)
\end{equation*}
as the number of approximating functions goes to infinity, $k \rightarrow \infty$.

Analogously to the fully parametric setup examined by \textcite{kline}, the plug-in estimator of the quadratic form $\hat{\theta}_{\text{PI}} := \hat{\beta}' A \hat{\beta}$ is biased, because
\begin{equation*}
    \E[\hat{\theta}_{\text{PI}} - \theta] = \text{trace} (A \var[\hat{\beta}]) = \sum_{i = 1} ^n B_{ii} \sigma_i ^2,
\end{equation*}
where $B_{ii} := \sum_{j = 1} ^n M_{ij} x_j' S_{xx} ^{-1} A S_{xx} ^{-1} \sum_{j = 1} ^n M_{ij} x_j$ measures the influence of the $i$th squared error $e_i ^2$ on $\hat{\theta}_{\text{PI}}$. In cases when $p / n \approx 0$, individual elements $B_{ii}$ are close to zero, and the bias is negligible. However, with increasing $p$, the bias is more pronounced and needs to be accounted for. The intuition is that $\hat{\beta}'A\hat{\beta}$ squares the estimation noise in $\hat{\beta}$ together with the signal, so it overstates the true dispersion $\beta'A\beta$. The leave-out correction estimates that noise term observation by observation and subtracts it, using for unit $i$ a residual variance that never draws on unit $i$'s own outcome.

We introduce some notation to describe our estimator. Collect $x_i$ in a matrix as $X := (x_1, \ldots, x_n)' \in \mathbb{R}^{n \times p}$. Then, $W := (X, P_k) \in \mathbb{R}^{n \times (p + k)}$ is a matrix containing the full set of regressors, with individual vectors $w_i$, and $W = (w_1, \ldots, w_n)'$. Define $$\hat{\gamma} := \left(\sum_{j = 1} ^n w_j w_j'\right)^{-1} \sum_{j = 1} ^n w_j y_j,$$ and the leave-one-out version of it, $$\hat{\gamma}_{-i} := \left(\sum_{j = 1} ^n w_j w_j' - w_i w_i'\right)^{-1} \left(\sum_{j = 1} ^n w_j y_j - w_i y_i \right).$$ Note that the first $p$ elements of $\hat{\gamma}$ are $\hat{\beta}$, and the remaining $k$ elements are basis coefficients from the function approximation. Similarly to \textcite{kline}, we estimate $\theta$ as
\begin{equation}\label{quad_est_main}
    \hat{\theta} := \hat{\beta}' A \hat{\beta} - \sum_{i = 1} ^n B_{ii} \hat{\sigma}_i ^2,
\end{equation}
where $\hat{\sigma}_i ^2$ is a leave-one-out estimator of the individual variance of $e_i$ defined as
\begin{equation}\label{sigma_hat}
    \hat{\sigma}_i ^2 := y_i \left(y_i - w_i'\hat{\gamma}_{-i}\right).
\end{equation}

We note that computing $\hat{\gamma}_{-i}$ for each $i = 1, \ldots, n$ is computationally costly in large-scale applications. To avoid this, we represent (\ref{sigma_hat}) as
\begin{equation}\label{sigma_elt}
    \hat{\sigma}_i ^2 = \frac{y_i \hat{e}_i}{M_{W, ii}},
\end{equation}
where $M_{W, ii}$ is $i$th diagonal element of the matrix $M_W := I_n - W (W'W)^{-1} W'$ projecting onto the complement of the column space spanned by $x_i$ and functions of $z_i$, and $\hat{e}_i := \sum_{j = 1} ^n M_{ij} (y_j - x_j \hat{\beta})$ are residuals.

Our estimator is general in the sense that it nests the linear leave-out estimator of KSS as a special case. The generality comes from enlarging the assumed function class from the linear class $\mathcal{F}_{\text{linear}}$ to the broader class $\mathcal{F}$ of nonlinear functions satisfying Assumption \ref{asum1}, so that $\mathcal{F}_\text{linear} \subseteq \mathcal{F}$.

To see this more clearly, note that our estimator $\hat{\theta}$ is \textit{exactly} unbiased if the underlying functional class is linear, $f \in \mathcal{F}_{\text{linear}}$. It follows from decomposition as in (\ref{decomp}) and noting that the $\mathcal{B}$ term disappears by the properties of the $M$ matrix that projects onto the complement of the space spanned by $z_i$ (where we set $P_k := P_d = (z_1, \ldots, z_n) \in \mathbb{R}^{n \times d}$), in which case $\sum_{i = 1} ^n \sum_{j = 1} ^n M_{ij} x_i z_j' \alpha = 0$. Under the condition that the unknown function $f(z_i) = z_i'\alpha$ is linear in $z$ (as in \textcite{kline}), as a result $\E[\hat{\beta}] = \beta$ holds, and we have
\begin{equation*}
    \begin{split}
        \E[\hat{\sigma}_i ^2] &= \E\left[y_i \left(y_i - x_i' \hat{\beta}_{-i} - p_k (z_i)' \hat{\alpha}_{-i}\right)\right] \\
        &= \E\left[\left(x_i' \beta + f(z_i) + e_i \right) \left(x_i' \beta + f(z_i) + e_i - x_i' \hat{\beta}_{-i} - p_k (z_i)' \hat{\alpha}_{-i}\right)\right] \\
        &= \E\left[\left(x_i' \beta + f(z_i) + e_i \right) \left(x_i' (\beta - \hat{\beta}_{-i}) + (f(z_i) - p_k (z_i)' \hat{\alpha}_{-i}) + e_i \right)\right] \\
        &= \left(x_i' \beta x_i' + f(z_i) x_i'\right) \E[\beta - \hat{\beta}_{-i}] + \left(x_i' \beta + f(z_i)\right) \E[f(z_i) - p_k (z_i)' \hat{\alpha}_{-i}] + \sigma_i ^2 \\
        &= \sigma_i ^2
    \end{split}
\end{equation*}
where we use the fact that $\E[e_i x_i' (\beta - \hat{\beta}_{-i})] = \E[e_i] \E[x_i' (\beta - \hat{\beta}_{-i})] = 0$, because $\hat{\beta}_{-i}$ is independent of $e_i$, and similarly $\E[e_i (f(z_i) - p_k (z_i)' \hat{\alpha}_{-i})] = \E[e_i] \E[f(z_i) - p_k (z_i)' \hat{\alpha}_{-i}] = 0$, because $\hat{\alpha}_{-i}$ is independent of $e_i$. 

However, for more general classes of functions $\mathcal{F}$, the exact unbiasedness of $\hat{\beta}$ and $\hat{\sigma}_i ^2$ does not hold. Instead, we study an asymptotic sequence in which the number of approximating functions diverges, $k \rightarrow \infty$, and show that the bias terms vanish under explicit conditions. Concretely,\footnote{For expositional purposes, we sometimes drop the Euclidean norm on the vector $\norm{v}$, and write asymptotic results as, say, $v = o_p(1)$ instead of $\norm{v} = o_p(1)$.}
\begin{equation*}
    \E[\hat{\beta}] - \beta = o(1),  \quad \E[\hat{\sigma}_i^2] - \sigma_i ^2 = o_p(1), \quad \E[\hat{\theta}] - \theta = o_p(1),
\end{equation*}
along sequences satisfying the assumptions below.

For these purposes, we assume that the unknown function $f(z)$ belongs to a class $\mathcal{F}$ of smooth functions that can be well approximated by linear combinations of the basis functions $p^1(z), \ldots, p^k(z)$.

\begin{assum}[Function class]\label{asum1}
We assume that $f \in \mathcal{F}$, where
    \begin{equation}
    \mathcal{F} := \left\{f: \min_{\alpha \in \mathbb{R}^k} \E\left[|f(z_i) - p_k (z_i)' \alpha|^2\right] \leq Ck^{-2 \alpha_f}, \alpha_f > 1 \right\}
\end{equation}
for some absolute constant $C < \infty$. Furthermore, $\sup_{f \in \mathcal{F}} f < M$ for some absolute constant $M < \infty$.
\end{assum}
This implies that
\begin{equation*}
    \sup_{f \in \mathcal{F}} \E\left[|f(z_i) - p_k (z_i)' \hat{\alpha}_{-i}|^2\right] = \mathcal{O}_p(k^{-2 \alpha_f}),
\end{equation*}
so the mean-square approximation error decays at rate $k^{-2\alpha_f}$ uniformly over $\mathcal{F}$.

The constant $\alpha_f$ governs how fast approximation bias decays.
To obtain consistency for the quadratic form studied here, we strengthen the usual semiparametric requirement from $\alpha_f > 0$ in \textcite{donald} and \textcite{cattaneoalt} to $\alpha_f > 1$: compared with slope-type estimands, quadratic forms require faster decay of approximation error.
For polynomial or spline series on compact support with $d_{\mathrm{cont}}$ continuous covariates, $\alpha_f = s_f / d_{\mathrm{cont}}$ where $s_f$ is the number of continuous derivatives \parencite{chen}, so the assumption reduces to $s_f > d_{\mathrm{cont}}$: the unknown function must have more continuous derivatives than the number of continuous covariates entering flexibly. 
Discrete covariates such as gender, occupation, or region do not count toward $d_{\mathrm{cont}}$ because they are absorbed by group interactions rather than by the polynomial basis.
In the wage applications of Section~\ref{sec:empirical}, the parsimonious sample has $d_{\mathrm{cont}} = 1$ (age), so $s_f > 1$ suffices---the wage--age profile need only be twice continuously differentiable. The Panel B specification adds employment, fixed assets, and intermediate inputs, giving $d_{\mathrm{cont}} \leq 4$ and requiring $s_f > 4$. Both conditions are mild for smooth economic relationships. The condition is sufficient but not shown to be necessary. This mirrors the elbow phenomenon in minimax estimation of quadratic functionals, where reaching the parametric rate requires more smoothness than a linear functional of the same nuisance would need \parencite{breunigchen22}. Even though the assumption on $\alpha_f$ is strengthened compared to the standard linear setups, it is equivalent to the assumption commonly placed in settings with nonlinear structure, as in, for example, \textcite{hirano} and \textcite{farrell}.

\begin{rem}\label{rem:scope}
The results cover a fixed number of continuous covariates and a fixed collection of discrete interaction groups, as in Section \ref{sec:empirical}. They do not cover a growing continuous dimension, a growing number of groups, or nuisance functions estimated by machine learning---random forests, boosted trees, or neural networks---whose approximating class is not the span of a known, pre-specified basis.\footnote{\textcite{twice} take that route on the same Portuguese data, replacing the worker and firm fixed effects with gradient-boosted trees fit to observable characteristics.}
\end{rem}

\section{Consistency}\label{sec:consistency}
In this section, we prove the consistency result for the proposed estimator $\hat{\theta}$. We study the asymptotic behavior of $\hat{\theta}$ assuming that $x_i$, $z_i$, and $A$ are sequences of constants so that the only source of randomness is $e_i$. We adopt the conditional perspective of \textcite{scheffe}, \textcite{searle}, and \textcite{kline}, treating $x_i$, $z_i$, and $A$ as fixed. This allows us to be agnostic about potential dependence between $x_i$, $z_i$, and $A$. Our analysis differs from \textcite{cattaneo18}, who consider sequences of random variables and condition on $z_i$ only. Limits are taken assuming the number of observations goes to infinity, $n \rightarrow \infty$, the number of approximating functions goes to infinity, $k \rightarrow \infty$ (so that the finite-sample bias of $\hat{\beta}$ and $\hat{\sigma}_i ^2$ is asymptotically of negligible order), and the dimensionality of $x_i$ goes to infinity, $p \rightarrow \infty$, to model the limited mobility bias that arises when workers move between only a few firms, so that worker and firm effects are estimated from limited variation. We make the following additional assumptions:
\begin{assum}[Data-generating process]\label{asum2}
    (i) $\max_i (\E[e_i ^4] + \sigma_i ^{-2}) = \mathcal{O}(1)$; (ii) there exists a $c < 1$ such that $\max_i P_{W, ii} < c$ for all n, where $P_{W, ii}:= 1 - M_{W, ii}$; (iii) $\max_i (x_i' \beta)^2 = \mathcal{O}(1)$.
\end{assum}

\noindent
Part \textit{(i)} excludes errors without a uniformly bounded fourth moment. The standardized Student-$t_5$ design in Section~\ref{sec:sim} is therefore a finite-sample tail-stress case, not a violation of the assumption. Together with Assumption \ref{asum1}, parts \textit{(ii)} and \textit{(iii)} imply that the leave-one-out estimator $\hat{\sigma}_i ^2$ is well-defined and has bounded variance. Part \textit{(ii)} also implies that $\frac{p + k}{n} \leq c < 1$ for all $n$. 

\begin{assum}[Growth rate conditions]\label{asum3}
    We assume that $k \rightarrow \infty$ and $p \rightarrow \infty$ as $n \rightarrow \infty$, and $k = o(n)$, $k = o(p)$, $n = o(k^{\alpha_f})$, and $p = \mathcal{O}(n)$.
\end{assum}

\noindent
This assumption places restrictions on the relationship between growth rates of the number of observations, linear covariates, and approximating functions. In particular, it reflects empirical practice in the AKM literature where the number of observations is of a much larger magnitude than the potential number of series terms. It also places a lower bound on the number of series terms, since $k$ must be large enough that the $n/k^{\alpha_f}$ term is small. Intuitively, the window $n^{1/ \alpha_f} \ll k \ll n$ balances zero asymptotic bias against a negligible added variance contribution. Relaxing the assumption to cover $k = \mathcal{O}(n)$ would be of independent theoretical interest, though of limited empirical relevance in the AKM setting. We conjecture that letting the number of series terms grow proportionally with the number of observations would add a further term to the asymptotic variance of $\hat{\theta}$, in line with the theory on linear functionals with growing dimension \parencite{cattaneo18}. We leave a proof to future work.

\paragraph{Shared rate restrictions for the proved results}
The results below all rely on the same core asymptotics: the series dimension grows ($k \rightarrow \infty$), total regressor dimension remains bounded away from the sample size through $\max_i P_{W,ii} < c < 1$, and approximation bias decays fast enough through $\alpha_f > 1$. The theorem-specific conditions then add the restrictions needed for each result: Lemma \ref{lem1} requires $\text{trace}(\tilde{A}^2)=o(1)$, Theorem \ref{chi_sq} adds the fixed-rank leverage condition $\max_i v_i'v_i=o(1)$, and Theorem \ref{big_r_norm} adds the growing-rank negligibility conditions on $\max_i \left((\tilde{w}_i' \gamma)^2 + (\check{w}_i' \gamma)^2\right)$ and on the leading eigenvalue share $\lambda_1^2 / \sum_{\ell = 1}^r \lambda_\ell^2$.
The tradeoff behind the fixed-dimensional series results is simple: $k$ must grow fast enough for approximation error to vanish, but not so fast that the combined regressor dimension $p+k$ drives leverage close to one.

Relative to linear leave-out estimation, replacing linear controls by a growing basis adds one additional task: one must show that the approximation terms from the unknown nuisance are negligible both for $\hat{\beta}$ and for the leave-out residual-variance correction. Lemma \ref{lem1} and Theorems \ref{chi_sq}--\ref{big_r_norm} establish exactly this under Assumptions \ref{asum1}--\ref{asum3}.

\begin{coro}\label{cor:fixedd}
Suppose the covariates entering the unknown function can be written as $z_i = (u_i', s_i)'$, where $u_i \in \mathbb{R}^{d_{\mathrm{cont}}}$ has fixed dimension, $s_i \in \{1, \ldots, G\}$ indexes a fixed collection of discrete groups, and
\begin{equation*}
    f(z_i) = \sum_{g = 1}^{G} \mathds{1}\{s_i = g\} f_g(u_i).
\end{equation*}
Assume the support of $u_i$ is compact and that, for the specified basis $\{p^j(u_i)\}_{j=1}^k$, each $f_g$ satisfies
\begin{equation*}
    \inf_{\alpha_g \in \mathbb{R}^k} \E\left[\left|f_g(u_i) - p_k(u_i)' \alpha_g\right|^2\right] \leq C k^{-2 s_f / d_{\mathrm{cont}}}
\end{equation*}
for a common constant $C < \infty$ and common smoothness index $s_f > d_{\mathrm{cont}}$, uniformly over $g$. Then the interacted basis $\{\mathds{1}\{s_i = g\} p^j(u_i)\}_{g,j}$ satisfies Assumption \ref{asum1} with $\alpha_f = s_f / d_{\mathrm{cont}} > 1$. If, in addition, Assumptions \ref{asum2} and \ref{asum3} hold together with the theorem-specific leverage and eigenvalue conditions, then the conclusions of Lemma \ref{lem1}, Theorem \ref{chi_sq}, and Theorem \ref{big_r_norm} apply to this interacted specification.
\end{coro}

Let $\tilde{A} := S_{xx}^{-1/2} A S_{xx}^{-1/2}$ and write its eigendecomposition as $\tilde{A} = Q \Lambda Q'$, where $\Lambda = \text{diag}(\lambda_1, \ldots, \lambda_r)$ collects the $r$ nonzero eigenvalues and $Q \in \mathbb{R}^{p \times r}$ the corresponding eigenvectors.

\begin{lem}\label{lem1}
    If Assumptions \ref{asum1}, \ref{asum2}, and \ref{asum3} hold, $A$ is positive semi-definite, $\theta = \beta' A \beta = \mathcal{O}(1)$, and $\text{trace} (\tilde{A}^2) = \sum_{\ell = 1} ^r \lambda_\ell^2 = o(1)$, then
    \begin{equation*}
        \hat{\theta} - \theta \overset{p}{\rightarrow} 0.
    \end{equation*}
\end{lem}

\section{Large-scale approximation}\label{sec:compute}
In this section we discuss an alternative estimator that allows for fast computation in typical large-scale applications such as those based on administrative linked employer-employee data. We follow \textcite{kline} by considering the random projection method of \textcite{achlioptas}, known as the Johnson--Lindenstrauss approximation (JLA), which names the estimator below. To describe it, fix $m \in \mathbb{N}$ and generate matrices $R_B, R_P \in \mathbb{R}^{m \times n}$ where each $(i, j)$ coordinate is a random draw from the following Rademacher distribution:
\begin{equation*}
    R_{\cdot, ij} = \begin{cases}
        +1 & \text{with probability 1/2}, \\
        -1 & \text{with probability 1/2}.
    \end{cases}
\end{equation*}
Decompose $A = 1/2 (A_1' A_2 + A_2'A_1)$ for $A_1, A_2 \in \mathbb{R}^{n \times p}$, where $A_1 = A_2$ if $A$ is positive semi-definite. Denote 
\begin{equation*}
    \hat{P}_{W, ii} := \frac{1}{m} \norm{R_P W S_{ww}^{-1} w_i}^2, \quad \hat{B}_{ii} := \frac{1}{m} \left(R_B A_1 S_{xx}^{-1} \sum_{j = 1} ^n M_{ij} x_j\right)' \left(R_B A_2 S_{xx}^{-1} \sum_{j = 1} ^n M_{ij} x_j\right).
\end{equation*}
The proposed estimator is then:
\begin{equation}\label{jla}
    \hat{\theta}_{\text{JLA}} := \hat{\beta}' A \hat{\beta} - \sum_{i = 1} ^n \hat{B}_{ii} \hat{\sigma}^2_{i, \text{JLA}}, \quad \hat{\sigma}^2_{i, \text{JLA}} := \frac{y_i (y_i - w_i' \hat{\gamma})}{1 - \hat{P}_{W, ii}} \left(1 - \frac{1}{m} \frac{3 \hat{P}_{W, ii} ^3 + \hat{P}_{W, ii} ^2}{1 - \hat{P}_{W, ii}}\right).
\end{equation}
As in \textcite{kline}, the term $\frac{1}{m} \frac{3 \hat{P}_{W, ii} ^3 + \hat{P}_{W, ii} ^2}{1 - \hat{P}_{W, ii}}$ removes a non-linearity bias from the approximation of $P_{W, ii}$ by $\hat{P}_{W, ii}$.

\begin{lem}\label{jla_lem}
    If Assumptions \ref{asum1}, \ref{asum2}, and \ref{asum3} are satisfied, $n / m^4 = o(1)$, $\var[\hat{\theta}]^{-1} = \mathcal{O}(n)$, and one of the following conditions hold, then $\var[\hat{\theta}]^{-1/2} (\hat{\theta}_{\text{JLA}} - \hat{\theta} - \mathrm{B}_m) = o_p(1)$ where $|\mathrm{B}_m| \leq \frac{1}{m} \sum_{i = 1} ^n P_{W, ii} ^2 |B_{ii}| \sigma_i ^2$:
    \begin{enumerate}[(i)]
        \item $A$ is positive semi-definite and $\E[\hat{\beta}' A \hat{\beta}] - \theta = \sum_{i = 1} ^n B_{ii} \sigma_i ^2 = \mathcal{O}(1)$.
        \item $A = 1/2 (A_1' A_2 + A_2' A_1)$ where $\theta_1 = \beta' A_1' A_1 \beta$ and $\theta_2 = \beta' A_2' A_2 \beta$ satisfy (i) and $\frac{\var[\hat{\theta}_1] \var[\hat{\theta}_2]}{n \var[\hat{\theta}]^2} = \mathcal{O}(1)$.
    \end{enumerate}
\end{lem}

\section{Limiting distributions}\label{sec:lim}

This section derives two limiting distribution results for $\hat{\theta}$. The first assumes that the rank $r$ of $A$ is fixed; the second allows $r$ to grow with $n$.

\subsection{Fixed rank}
First, we derive the limiting distribution of $\hat{\theta}$ for fixed $r$. Using the eigendecomposition $\tilde{A} = Q \Lambda Q'$ introduced before Lemma \ref{lem1}, we represent the estimator as
\begin{equation*}
    \hat{\theta} = \sum_{\ell = 1} ^r \lambda_\ell \left(\hat{b}_\ell ^2 - \widehat{\var}[\hat{b}_\ell]\right),
\end{equation*}
where $\hat{b} := \sum_{i = 1} ^n v_i y_i$, $\widehat{\var}[\hat{b}] := \sum_{i = 1} ^n v_i v_i' \hat{\sigma}_i ^2$, and $v_i := Q' S_{xx}^{-1/2} \sum_{j = 1} ^n M_{ij} x_j$. 

The next theorem gives the fixed-rank limit for $\hat{\theta}$ and consistency of $\widehat{\var}[\hat{b}]$. 

\begin{thm}\label{chi_sq}
    If Assumptions \ref{asum1}, \ref{asum2}, and \ref{asum3} hold, $r$ is fixed, and $\max_i v_i' v_i = o(1)$, then
    \begin{enumerate}[(i)]
        \item $\var[\hat{b}]^{-1/2} (\hat{b} - b) \overset{d}{\rightarrow} \mathcal{N}(0, I_r)$, where $b := Q' S_{xx}^{1/2} \beta$,
        \item $\var[\hat{b}]^{-1} \widehat{\var}[\hat{b}] \overset{p}{\rightarrow} I_r$,
        \item $\hat{\theta} = \sum_{\ell = 1} ^r \lambda_\ell \left( \hat{b}_\ell ^2 - \var[\hat{b}_\ell] \right) + o_p (\var[\hat{\theta}]^{1/2})$.
    \end{enumerate}
\end{thm}

The growing-rank case is more relevant for AKM-type applications, so we turn to that setting next.

\subsection{Growing rank}

Define $\tilde{w}_i := \check{A} S_{ww}^{-1} w_i$, where $\check{A} := \begin{pmatrix} A & 0 \\ 0 & 0 \end{pmatrix}$ and $\gamma := \begin{pmatrix} \beta & \alpha \end{pmatrix}'$. Then
\begin{equation*}
    \theta = \gamma' \check{A} \gamma = \gamma' S_{ww} S_{ww}^{-1} \check{A} \gamma = \sum_{i = 1} ^n \gamma' w_i \tilde{w}_i' \gamma
\end{equation*}
and the leave-one-out estimator can be written as
\begin{equation*}
    \hat{\theta} = \sum_{i = 1} ^n y_i \tilde{w}_i' \hat{\gamma}_{-i},
\end{equation*}
where the approximation error from replacing $\gamma' w_i$ by $y_i$ is controlled by Assumption \ref{asum1} and \ref{asum3}. A direct algebraic rearrangement yields
\begin{equation*}
    \begin{split}
        \hat{\theta} =  \sum_{i = 1} ^n \sum_{\ell \neq i} C_{i \ell} y_i y_\ell,
    \end{split}
\end{equation*}
where $C_{i\ell} := B_{W, i\ell} - 2^{-1} M_{W, i\ell} \left(B_{W, ii} M_{W, ii}^{-1} + B_{W, \ell \ell} M_{W, \ell \ell}^{-1}\right)$. Thus $\hat{\theta}$ is a second-order $U$-statistic with sample-dependent kernel $C_{i\ell}$. Theorem \ref{big_r_norm} combines this representation with Appendix Lemma \ref{solv}; relative to \textcite{kline}, the only additional term is the series approximation error, which is negligible under Assumption \ref{asum1} and \ref{asum3}.

Define also
\begin{equation*}
    \check{w}_i := \sum_{\ell = 1} ^n M_{W, i\ell} \frac{B_{W, \ell \ell}}{1 - P_{W,\ell \ell}} w_\ell.
\end{equation*}

\begin{thm}\label{big_r_norm}
    If Assumption \ref{asum1}, \ref{asum2}, and \ref{asum3} hold, and the following conditions are satisfied,
    \begin{equation*}
        \text{(i)} \, \var[\hat{\theta}]^{-1} \max_{i} \left((\tilde{w}_i' \gamma)^2 + (\check{w}_i' \gamma)^2\right) = o(1), \quad \text{(ii)} \, \frac{\lambda_1 ^2}{\displaystyle \sum_{\ell = 1} ^r \lambda_\ell ^2} = o(1),
    \end{equation*}
    then $\var[\hat{\theta}]^{-1/2} (\hat{\theta} - \theta) \overset{d}{\rightarrow} \mathcal{N}(0, 1)$. 
\end{thm}

\section{Simulation study}\label{sec:sim}
We use a Monte Carlo study to evaluate the estimator under controlled conditions that isolate the problems described in the theory. 
In the baseline design, the covariates $z_i$ are drawn uniformly on a fixed cube, the linear regressors $x_i$ (so called because they enter the outcome linearly) are generated as nonlinear functions of $z_i$, and the quadratic form of interest is $\theta = \beta' A \beta$, where $A$ selects the first two coordinates of $\beta$, so it measures dispersion in two fixed effects.
Unless otherwise noted, the nonlinear designs use $n=500$, nuisance dimension $d=4$, and $p=90$ linear regressors.
The rows are constructed as comparative statics: within a panel, the simulation holds the underlying random design fixed whenever dimensions permit and changes only the feature named in the row---the nonlinearity of $f(z)$, the error distribution, the nuisance dimension, the number of linear regressors, leverage, or sample size.

More specifically, each scenario first draws a fixed design $(z_i, x_i)_{i=1}^n$ and coefficient vector $\beta$, then holds that design fixed across Monte Carlo replications while redrawing only the error term $e_i$. 
The first panel of Table \ref{tab:simulation_results} uses the same $z_i$, $x_i$, $\beta$, and error draws and changes only the functional form of $f$.
The second panel keeps the strongly nonlinear radial design as the reference case; the heteroskedastic and heavy-tailed rows change only the error process, while the remaining rows change one of the other named features using nested random draws where possible.
The nuisance covariates satisfy $z_i \sim \text{Unif}([-1,1]^d)$. 
Before standardization, the linear regressors are
\[
    x_{ij}^{*} = \exp\{0.30 \norm{z_i}_2 + 0.50 \eta_{ij} + 0.10 \nu_j\},
    \qquad \eta_{ij}, \nu_j \sim \mathcal{N}(0,1),
\]
where $\eta_{ij}$ is an idiosyncratic shock drawn separately for each observation and regressor, and $\nu_j$ is a single shock shared by every observation in regressor column $j$; each column of $x_i^*$ is centered and scaled.
The shared term $0.30\norm{z_i}_2$ correlates every regressor with the nuisance covariate, entangling $x_i$ and $z_i$ just as the bias term $\mathcal{B}$ in (\ref{decomp}) requires: without that correlation, a linear regression of $y$ on $x$ alone would already be unbiased for $\beta$, leaving nothing for the semiparametric correction to fix.
The nuisance signal is also centered and scaled before entering the outcome.
We use three functional forms:
\[
    f_{\mathrm{linear}}(z_i) = 0.8 z_{i1} - 0.5 z_{i2} + 0.35 z_{i3} - 0.2 z_{i4}, \quad
    f_{\mathrm{cubic}}(z_i) = 0.8 z_{i1}^3 - 0.5 z_{i2}^2 + 0.3 z_{i3}z_{i4},
\]
and the strongly nonlinear radial design
\[
    f_{\mathrm{radial}}(z_i) = \norm{z_i}_2^7.
\]
The radial seventh-power function is the main nonlinear stress test: it differs sharply from any linear control because its slope rises quickly as observations move away from the center of the cube. 
The error term is standard Gaussian, heteroskedastic Gaussian with scale $0.5+0.5|x_{i1}|$ normalized by its fixed-design root mean square, or a variance-one Student-$t$ distribution with five degrees of freedom.\footnote{These are population or fixed-design normalizations; we do not divide each Monte Carlo draw by its realized sample standard deviation.}
The many-regressor, higher-dimension, and larger-sample designs each raise one of $p$, $d$, or $n$ above its baseline value, and the high-leverage design inflates a small share of extreme observations to create thinly supported, high-leverage points.\footnote{%
In Table \ref{tab:simulation_results}, $p=90$ in all scenarios except the many-regressor design, where $p=300$.
The nuisance dimension is $d=4$ in all scenarios except the higher-dimension design, where $d=8$.
The sample size is $n=500$ in all scenarios except the larger-sample design, where $n=1000$.
In the high-leverage design, the regressors for the 5\% of observations with the largest $|z_{i1}|$---those farthest from the center of the cube along the first nuisance coordinate---are additionally scaled up by a factor of six in the first ten regressor columns (or all columns, if $p<10$), creating a small number of extreme, thinly supported points that drive up their own leverage $P_{W,ii}$.}

We compare the plug-in estimator with leave-out estimators that use complete total-degree polynomial bases of degree 1, 3, and 5, including all monomials whose total degree does not exceed the stated degree. 
Cubic polynomials are common in applied work \parencite{card18}; the degree-5 basis lets us check whether going beyond cubic matters. 
Table \ref{tab:simulation_results} reports the subset of designs that cover the main cases in the paper: approximation error and nonlinearity, robustness to heteroskedasticity and heavy tails, and finite-sample performance as nuisance dimension, the number of regressors, leverage, and sample size change one at a time. 
The first four columns report signed Monte Carlo bias for the plug-in (PI) and leave-out (LOO) estimators at each basis degree.
The fifth column reports RMSE for LOO at degree~5. Appendix Table \ref{tab:simulation_coverage} reports nominal 95\% coverage for the same estimator separately, with confidence intervals based on the estimated variance of the quadratic form computed from the leave-one-out variances $\hat{\sigma}_i^2$.

\begin{table}[t]
\centering
\caption{Simulation study: bias and finite-sample performance}
\vspace*{0.5mm}
\input{tables/simulation_results}
\note[Notes]{%
Columns PI(1), LOO(1), LOO(3), and LOO(5) report signed Monte Carlo bias, $\E[\hat{\theta} - \theta]$, for the plug-in (PI) estimator with a degree-1 basis and the leave-out (LOO) estimator with degree-1, degree-3, and degree-5 bases, respectively. Positive entries indicate overestimation of the quadratic form.
RMSE(5) reports the root mean squared error of the degree-5 LOO estimator. Bias and RMSE are reported in the units of the quadratic-form target. Appendix Table \ref{tab:simulation_coverage} reports the corresponding nominal 95\% coverage diagnostic. Intervals use the estimated variance of the quadratic form computed from the leave-one-out residual variances $\hat{\sigma}_i^2$.
The first block varies only the functional form of $f(z)$. The second block holds the strong radial design fixed and changes only the feature named in the row, relative to the strong-nonlinearity row.
Each estimable design uses 5,000 Monte Carlo replications. \emph{NA} means that the complete basis is not estimable: in the higher-dimension row, $d=8$ gives the degree-5 basis 1,287 columns, so $p+k>n$.}
\label{tab:simulation_results}
\end{table}

\paragraph{Approximation error and the smoothness condition}
Assumption \ref{asum1} predicts that series approximation error decays at rate $k^{-2\alpha_f}$, so higher-degree bases should reduce bias most when $f$ departs strongly from linearity. 
The simulations confirm this. 
In the strong radial design, the bias falls in magnitude from about $-0.04$ for PI(1) and LOO(1) to $-0.02$ for degree 3 and $-0.01$ for degree 5. 
Under high leverage, the same pattern holds: the bias falls from about $-0.02$ for PI(1) and $-0.03$ for LOO(1) to zero after adding higher-degree basis terms. 
By contrast, in the linear benchmark and the mildly nonlinear cubic design, the bias is already close to zero, with the largest entry about $0.01$. 
The first panel of Table \ref{tab:simulation_results} should therefore be read as evidence on where the smoothness condition matters in finite samples: adding basis terms removes systematic misspecification bias when nonlinearity is large, and costs little when it is not.

\paragraph{Heteroskedasticity and error tails}
The leave-out correction in (\ref{quad_est_main}) uses observation-specific residual variances $\hat{\sigma}_i^2$, so it depends on the moment conditions in Assumption \ref{asum2}\textit{(i)}. 
The heteroskedastic and heavy-tailed designs show why the table reports both the bias and RMSE. 
Because these rows now keep the same radial design as the homoskedastic strong-nonlinearity row, the low-degree bias is similar: about $-0.03$ to $-0.04$ before the higher-degree basis removes most of it. 
Degree-5 RMSE is about $0.05$ under both heteroskedasticity and Student-$t_5$ errors. 
These designs therefore create dispersion that is not captured by the bias alone; adding basis terms removes the approximation component, but residual-variance estimation remains a finite-sample issue.

\paragraph{Leverage and many regressors}
Leverage measures how much an observation's own outcome drives its own fitted value; it is high when a worker or firm effect is pinned down by very little data, as with rarely observed workers or small firms.
Assumption \ref{asum2}\textit{(ii)} requires $\max_i P_{W,ii} < c < 1$, bounding individual leverage away from one as the combined regressor dimension $p + k$ grows.
The many-regressor and high-leverage designs make this condition hard to satisfy. 
Point estimation remains reliable: in the many-regressor design, the degree-5 leave-out estimator has bias around $-0.01$ and RMSE of about $0.06$. 
Appendix Table \ref{tab:simulation_coverage} reports the corresponding coverage diagnostic: nominal 95\% coverage is about $0.86$ in the high-leverage design, $0.97$ in the many-regressor design, $0.90$ under heteroskedasticity, and $0.91$ under Student-$t_5$ errors.\footnote{%
Every design here under-covers except many regressors, which over-covers: the estimated variance used to build the interval is conservative when $p$ is large relative to $n$, so the interval is too wide rather than too narrow. Only the under-coverage cases are the concerning direction for applied inference.}
In the higher-dimension design, degree 3 remains feasible, with bias around $-0.02$, but the complete degree-5 basis is infeasible and is reported as ``NA.''
The pattern points to residual-variance estimation and leverage, rather than point-estimate bias alone, as the main finite-sample problem.
Larger samples improve point accuracy: bias is near zero, RMSE falls to about $0.03$, and intervals again cover close to $95\%$ of draws.

\paragraph{Summary}
For applied work, adding basis terms is most valuable when nonlinear misspecification is large, especially under high leverage. 
When those features are absent, additional flexibility costs little but gains little. 
The point estimator is reliable across most designs.

\paragraph{AKM-calibrated Monte Carlo}
We adapt the calibration logic in Section 9.5 of \textcite{kline} to our application.\footnote{We do not reproduce their first-differenced mover-pair design.}
We hold fixed a small connected mover graph (a set of workers and firms linked by job moves, within which worker and firm fixed effects are jointly identified) with 13,939 person-year observations, 2,393 workers, 639 firms, and 6,366 worker--firm matches, using only gender and age as controls.
The OLS worker and firm effects on this graph are rescaled so that their variances match the leave-out estimates.
Each match shock is drawn from a standardized Student-$t$ distribution with five degrees of freedom, with a heteroskedastic variance. 
Following \textcite{kline}, we let that variance depend on the firm-variance diagonal $B_{gg}$, the leverage $P_{gg}$, and the log employment of the worker's current and other firms through an exponential function fit to the leave-out variance estimates.\footnote{%
The exponential model is fit by nonlinear least squares in levels to all finite leave-out variance estimates, using the positive estimates only to initialize. 
Because the KSS regression premultiplies each match mean by the square root of its person-year count, the fitted variance is in weighted-match units; we divide it by the match count before assigning a common shock to the match's person-year observations.}
The design crosses linear and degree-5 nuisance DGPs with linear and degree-5 fitted nuisance specifications, holding firm assignments, covariates, and year assignments fixed across 1,000 draws.

\begin{table}[h!]
\centering
\caption{Portuguese-graph calibrated Monte Carlo for firm-effect variance}
\label{tab:simulation_akm_calibrated}
\vspace*{0.5mm}
\input{tables/data/gakm_akm_calibrated/simulation_results_akm_calibrated}
\note[Notes]{%
The true value is the calibrated firm-effect variance, and all bias entries are percentages relative to that value. KSS bias (MC SD) reports the relative bias of the KSS point estimate, followed in parentheses by the Monte Carlo standard deviation of the relative estimation error. The parenthesized Monte Carlo standard deviation measures dispersion across replications, not a statistical standard error. HO bias and PI bias report the relative bias of the homoskedasticity-only correction and the plug-in estimator, respectively. Linear and Flexible denote the linear and degree-5 nuisance fits. Each design uses 1,000 draws.}
\end{table}

Table \ref{tab:simulation_akm_calibrated} evaluates the point correction. KSS relative bias is $0.06$--$0.18\%$ across the four designs, compared with $1.24$--$1.60\%$ for the homoskedastic correction and $1.90$--$2.37\%$ for the plug-in estimator. 
The KSS Monte Carlo standard deviation is $0.48$--$0.62\%$ of the true value. 
The KSS bias is largest in the Nonlinear DGP, Linear-fit design---the case where the true relationship is nonlinear but the fitted specification is linear---but it stays tiny throughout, because this parsimonious calibration carries only an age profile. 
The exercise therefore checks only that the leave-out point correction stays nearly unbiased on a realistic mobility graph. 
It does not speak to the functional-form sensitivity found in the application, which operates through the firm inputs that this design excludes.

This exercise checks point-estimate bias, not inference, and the feasible standard-error calculation is not validated by it: the calculation excludes about half of the match-level leave-out variance estimates, and the surviving estimates inflate the standard error six- to eightfold relative to the Monte Carlo standard deviation---so far that the nominal 95\% interval covers in all 1,000 draws of every design, the signature of an interval too wide to be informative rather than a validated one. 
We treat those interval results as diagnostics, not as evidence for cluster-level inference.

\section{Empirical application}\label{sec:empirical}

We use Portuguese linked employer--employee data to study how alternative control specifications change the worker--firm wage decomposition.
Adding worker and firm-input controls lowers all three variance components. Allowing those controls to enter nonlinearly produces smaller movements in firm variance and sorting, while the worker component moves more across specifications.

\subsection{Data and sample selection}
Our empirical analysis combines two Portuguese sources of administrative data. 
The matched employer–employee dataset \emph{Quadros de Pessoal} (QP) reports worker characteristics (age, gender, occupation, qualification, education, contract type, hours, compensation) and firm characteristics (location, industry, employment) for the universe of private-sector firms. 
The firm accounts dataset \emph{Central de Balanços} (CB) provides annual balance-sheet and income-statement information, including fixed assets and intermediate inputs, for non-financial corporations. 
We link QP and CB at the firm level to obtain a panel with worker characteristics and firm inputs, following the same linkage of Portuguese wage records to firm accounts used by \textcite{cardcardosokline16}.

We restrict our sample to the period 2008--2018 and to full-time workers aged 20--65. 
The dependent variable is the log of individual hourly wages.
Following the standard AKM literature, we decompose the cross-sectional variance of log wages into the variance of worker fixed effects $(\sigma_{\alpha}^{2})$, the variance of firm fixed effects $(\sigma_{\psi}^{2})$, and their covariance $\cov(\alpha,\psi)$, which captures sorting.
Detailed sample descriptives, estimation-variable coverage, and sample-flow statistics are reported in Appendix Tables \ref{tab:descriptives}, \ref{tab:descriptive_coverage}, and \ref{tab:desc_cc}.

\subsection{Estimation}

We study two versions of this empirical exercise. 
The first is a parsimonious specification that uses only worker-side controls: gender and age. 
The second specification adds worker education and qualification together with log firm employment, $\log(1+\text{fixed assets})$, and $\log(1+\text{intermediate inputs})$. 
About $7.6\%$ of firm-years report exactly zero fixed assets, so the two monetary inputs enter as $\log(1+x)$ rather than $\log(x)$ to retain these economically meaningful zeros.\footnote{The transformation $\log(1+x)$ is not invariant to the units in which $x$ is measured \parencite{chenroth}: rescaling the monetary inputs changes where the accounting zeros sit relative to the positive values. This bears on the linear specifications more than on the flexible ones, which approximate an arbitrary smooth function of each input and are therefore nearly unaffected by the choice of units. Fixed assets and intermediate inputs are held at their recorded scale throughout, and both enter only as nuisance controls whose coefficients we do not interpret.} The transformations also keep the raw monetary scales from dominating the polynomial basis, and negative accounting entries are treated as missing.
In the empirical tables, we label these two exercises as Panel A (parsimonious controls) and Panel B (worker and firm-input controls). 
Because Panel B requires non-missing firm controls, the main empirical tables estimate both panels on the Panel B non-missing \emph{leave-match-out} set: the largest connected worker--firm set (linked by job moves, so that worker and firm fixed effects are jointly identified) that supports the match-level correction.\footnote{The source panel spans 2008--2018, but fixed assets are unavailable in 2008--2009. After imposing the Panel B non-missing firm-control requirement, the common Panel A--Panel B estimation sample is therefore effectively drawn from 2010--2018; Appendix Table \ref{tab:firm_controls_variation} reports the year-by-year joint availability of the firm controls.}
This keeps the main cross-panel comparison about controls rather than sample composition; Appendix Table \ref{tab:desc_cc} reports the larger parsimonious candidate sample as well as the common main estimation sample.

Within each of these two control sets, we estimate the same nested five-step specification ladder, progressively relaxing the restrictions on how covariates enter the wage equation: 
(i) a pure AKM baseline with no observed controls beyond the fixed effects; 
(ii) linear additive controls, meaning common quadratic and cubic age terms centered at 50, with the common linear age direction excluded because it is not identified with worker and year fixed effects, together with education, qualification, and firm inputs entering linearly and additively;
(iii) heterogeneous linear controls, which add demographic-group deviations in the age profile and group-specific slopes on the firm inputs;
(iv) nonlinear homogeneous controls, which replace those terms with a common degree-5 polynomial basis; and
(v) the full model, which interacts that basis with demographic groups.
The ladder separates the roles of adding observables, allowing slope heterogeneity, and allowing nonlinearity.
Because the linear additive rung already includes quadratic and cubic age terms, the move to the degree-5 basis adds higher-order age terms and richer functions of the firm inputs, including their multivariate interactions.
The polynomial basis is interacted with a fully saturated set of worker categorical variables, so the full model allows group-specific nonlinear effects of each continuous covariate: in Panel A the interaction groups are defined by gender, while in Panel B they are defined by gender $\times$ education $\times$ qualification.

We use the polynomial basis for the main ladder and re-estimate its two nonlinear rungs with two alternatives. Probabilists' Hermite polynomials provide a different coordinate system for degree-5 polynomial terms.\footnote{A low-order term---such as the linear age profile, which is excluded because it is not separately identified from the fixed effects---can be absorbed differently by different bases: the fixed-effect normalization pins down which specific combination of basis terms counts as part of the fixed effects rather than the control function, so changing the basis (monomial versus Hermite) changes how that excluded, collinear variation is split between the two. The resulting decomposition can therefore differ slightly from the monomial specification, even though both bases approximate the same underlying function.} Additive cubic B-splines provide a different functional form with quantile-spaced knots. We choose the spline counts so that each alternative has the same realized control dimension as the corresponding polynomial model: $k=4$ and $9$ in Panel A, and $k=129$ and $3{,}004$ in Panel B for the common and group-specific nonlinear specifications. Appendix Table \ref{tab:empirical_basis_robustness} compares the resulting full-regressor corrected decompositions.

\paragraph{Mapping to theory}
To map the empirical specification into the model in Section \ref{sec:setup}, let $z_i^{\text{cont}}$ denote the continuous covariates that enter flexibly and let $s_i \in \{1, \ldots, G\}$ denote the worker categorical group that defines the interaction cell.
In Panel A, $z_i^{\text{cont}}$ contains age and $s_i$ is gender. 
In Panel B, $z_i^{\text{cont}}$ contains age, log employment, $\log(1+\text{fixed assets})$, and $\log(1+\text{intermediate inputs})$, while $s_i$ indexes the fixed collection of gender $\times$ education $\times$ qualification cells. 
Writing $x_i$ for the regressors that remain outside the flexible nuisance term---the worker and firm fixed effects, year effects, and any controls that enter linearly at a given specification---the full specification can be written as
\begin{equation*}
    y_i = x_i'\beta + f_{s_i}(z_i^{\text{cont}}) + e_i
    = x_i'\beta + \sum_{s = 1}^{G} \mathds{1}\{s_i = s\}\, f_s(z_i^{\text{cont}}) + e_i.
\end{equation*}
This is a special case of \eqref{main}: the unknown function is simply evaluated on the enlarged argument $(z_i^{\text{cont}\prime}, s_i)'$, so the nuisance is allowed to be a different smooth function in each worker group.

Our implementation centers and scales the continuous inputs, approximates each group-specific function $f_s$ with the same polynomial family, and lets the coefficients vary across groups.
Concretely, if $\{p^j(z_i^{\text{cont}})\}_{j=1}^k$ is the polynomial basis, then the interacted series uses the terms $\{\mathds{1}\{s_i = s\} p^j(z_i^{\text{cont}})\}_{s,j}$. Equivalently, the model fits a separate smooth profile of age in Panel A, and separate smooth profiles of age and the three transformed firm inputs in Panel B, within each worker category.

For the theory, these interactions do not increase the continuous dimension of the approximation problem. 
The indicators $\mathds{1}\{s_i = s\}$ split the sample into a fixed number of groups, but smoothness is still imposed only with respect to the continuous arguments in $z_i^{\text{cont}}$. 
Therefore Assumption \ref{asum1} applies within each group at rate $\alpha_f = s_f / d_{\text{cont}}$, where $d_{\text{cont}}$ is the number of continuous covariates entering the nuisance and $s_f$ is the number of their continuous derivatives. 
This is precisely the fixed-dimensional interacted-series case covered by Corollary \ref{cor:fixedd}: $d_{\text{cont}} = 1$ in Panel A and $d_{\text{cont}} = 4$ in Panel B.

The empirical point correction uses the same augmented regressors as the theory: worker effects, firm effects, year effects, and every control-function term enter one joint sparse regression, and the correction is computed from that full design rather than from a residualized fixed-effect regression. The usual two-step alternative --- preadjusting wages for covariates and then applying the leave-out correction to the residualized outcome --- introduces a higher-order bias because the first-step estimation error enters both the outcome and the residual used in the correction \parencite{kline2024}; the joint design avoids this by estimating the control function alongside the fixed effects rather than partialling it out beforehand. The implementation differs in its deletion and dependence structure. The theory deletes independent observations one at a time. The application instead uses 20 outcome-independent cross-fit folds, treating an entire worker--firm match as one cluster for movers (workers observed at more than one firm) and individual person-years as singleton clusters for stayers (workers observed at only one firm); deleting a stayer's only match would remove that worker's column from the training design. This full-regressor cluster-fold correction aligns the empirical point estimand with the augmented-regressor construction in Section \ref{sec:finite}. Its cluster structure and fold-shared residuals remain outside the observation-level inference theory, so the results below report point estimates only, without theorem-based standard errors.

The interaction groups are defined from worker-side categorical variables, not firm identifiers. 
This matters for identification. 
Interacting the series with firm identifiers would produce columns that are nearly collinear with the firm fixed effects already contained in $X$, undermining the full-rank condition on $S_{xx} = X'MX$ required by Assumption \ref{asum2}.\footnote{%
Age is normalized at 50 throughout, and the common linear age term is excluded because it is not identified in the presence of worker and year fixed effects.
All variance components in Tables \ref{tab:empirical_results} and \ref{tab:empirical_results_delta} are reported as shares of total log-wage variance on the same leave-match-out estimation set.}

The firm-side controls that distinguish Panel B from Panel A are employment, fixed assets, and intermediate inputs. Appendix Table \ref{tab:firm_controls_variation} reports their within-firm variation and missingness in the cleaned firm-year panel.
On their estimation scales, their within-firm variance shares are about 7.9\% for employment, 17.8\% for fixed assets, and 9.8\% for intermediate inputs, so most of their variation is cross-sectional. On the raw headcount and euro scales the corresponding within-firm shares are only 5.8\%, 1.6\%, and 5.0\%: differences in firm size across firms dominate the year-to-year movements within a firm. Taking logs removes that size gap and raises the within-firm share of fixed assets more than tenfold, from 1.6\% to 17.8\%.
Specification diagnostics for each specification are reported in Appendix Table \ref{tab:empirical_diagnostics}.

\subsection{Results}
Table \ref{tab:empirical_results} reports the full-regressor bias-corrected variance components for each specification on the common Panel B leave-match-out sample. We report point estimates and do not attach standard errors because the cluster-fold correction lies outside the observation-level inference theory developed above. Cross-specification differences are descriptive.
Table \ref{tab:empirical_results_residualized} reports $\var(\tilde y)/\var(y)$, the ratio of residualized-wage variance to total wage variance after partialling out observables and year effects; values above one signal that near-collinear controls are amplifying rather than absorbing variance.
Appendix Table \ref{tab:empirical_results_delta} decomposes the within-panel movements into stepwise changes along two paths---one adding heterogeneity first, the other nonlinearity first---both arriving at the full model.
Appendix Table \ref{tab:empirical_basis_robustness} compares the nonlinear estimates across polynomial, Hermite, and cubic B-spline bases at the same realized control dimension.

\begin{table}[t]
\centering
\caption{Wage variance decomposition}
\label{tab:empirical_results}
\vspace*{1mm}
\input{tables/empirical_results}
\note[Notes]{%
Entries report full-regressor cluster-fold bias-corrected variance components as shares of total log-wage variance on the common Panel B leave-match-out estimation sample: worker-effect variance ($\sigma_\alpha^2/\var(y)$), firm-effect variance ($\sigma_\psi^2/\var(y)$), and sorting covariance ($\cov(\alpha,\psi)/\var(y)$). The table reports the covariance once; its contribution to the variance decomposition is twice this term. The correction uses 20 outcome-independent folds, treating mover matches as clusters and stayer person-years as singleton clusters.
Panel A's control set is gender and age only. Panel B adds education, qualification, log employment, $\log(1+\text{fixed assets})$, and $\log(1+\text{intermediate inputs})$. Panel A is rerun on that same common Panel B leave-match-out sample, so the panels differ in controls rather than sample composition. Age is normalized at 50; the other continuous inputs are centered at their sample means, and all continuous inputs are scaled before the polynomial basis is constructed.
The specification sequence is defined within each panel: \emph{Linear additive controls} uses common quadratic and cubic age terms centered at 50, excludes the common linear age direction because it is not identified with worker and year fixed effects, and uses linear firm inputs in Panel B; \emph{Heterogeneous linear controls} allows group-specific age deviations and, in Panel B, firm-input slopes; \emph{Nonlinear homogeneous controls} uses a common degree-5 polynomial basis; and \emph{Full model} interacts that basis with worker groups. The table reports point estimates only; no statistical standard errors or formal tests are reported. Bias-corrected components need not sum to one.}
\end{table}

\begin{table}[ht]
\centering
\caption{Variance remaining after partialling out observables}
\label{tab:empirical_results_residualized}
\vspace*{1mm}
\input{tables/empirical_results_residualized_variance}
\note[Notes]{%
Entries report $\var(\tilde y)/\var(y)$ on the common Panel B leave-match-out estimation sample used in Table \ref{tab:empirical_results}, where $y$ is log hourly wage. Here $\tilde y$ is the AKM input outcome after partialling out observables and year effects, before the remaining variation is decomposed into worker, firm, and sorting components.
This is a descriptive accounting statistic. It is not a bounded share of variance explained by the controls or the full residual from the joint worker--firm--controls--year regression. Values above one are possible when partialling out observables and year effects amplifies variance because of collinearity. No statistical standard errors are reported.}
\end{table}

\paragraph{Panel A: parsimonious controls}
Panel A is stable across the specification ladder.
Worker variance ranges from $0.533$ to $0.562$, firm variance from $0.136$ to $0.144$, and sorting from $0.074$ to $0.080$.
The residualized-wage variance ratio ranges from $0.956$ to $0.999$.
Thus neither the standard age profile nor its heterogeneous and nonlinear extensions materially changes the decomposition on the common Panel B sample.

\paragraph{Panel B: worker and firm-input controls}
Adding the common quadratic and cubic age terms, education, qualification, and the transformed firm inputs changes the decomposition already in the linear additive specification.
Relative to the AKM baseline, the additive specification lowers worker variance from $0.551$ to $0.474$, firm variance from $0.144$ to $0.121$, and sorting from $0.080$ to $0.047$.
Allowing heterogeneous linear slopes lowers worker variance further to $0.451$, while firm variance and sorting move only slightly, to $0.120$ and $0.045$.

A common degree-5 polynomial raises worker variance relative to the additive model, from $0.474$ to $0.491$, while slightly lowering firm variance to $0.120$ and sorting to $0.041$.
Interacting the polynomial basis with worker groups yields worker variance of $0.490$, firm variance of $0.117$, and sorting of $0.042$.
The residualized-wage variance ratios are $0.827$, $0.798$, $0.832$, and $0.827$ for the additive, heterogeneous-linear, common-polynomial, and full specifications, respectively; the heterogeneous-linear specification has the smallest ratio.

\paragraph{What changes across specifications}
Employment, fixed assets, and intermediate inputs are economically meaningful firm-side controls, but most of their transformed variation is cross-sectional.
The largest movement in Panel B comes from adding observed controls. Relative to the baseline, the linear additive specification lowers worker variance by $0.077$, firm variance by $0.022$, and sorting by $0.034$.
Relative to that additive specification, the common polynomial raises worker variance by $0.018$ while lowering firm variance by $0.002$ and sorting by $0.006$.
Adding group interactions to the polynomial changes the three components by only $-0.001$, $-0.003$, and $+0.001$.
The ladder therefore distinguishes a substantial observables adjustment from the smaller reallocation generated by functional-form flexibility.

\paragraph{Robustness to the choice of basis}
The Panel B polynomial and Hermite estimates are nearly identical. Their common specifications give worker variance of $0.491$ and $0.490$, respectively, firm variance of $0.120$ under both, and sorting of $0.041$ under both. Their group-specific specifications also agree closely. The cubic B-spline gives a common decomposition of $0.492$, $0.116$, and $0.040$, and a group-specific decomposition of $0.474$, $0.114$, and $0.041$. Across the three group-specific bases, firm variance therefore lies between $0.114$ and $0.117$, and sorting between $0.041$ and $0.042$; worker variance ranges from $0.474$ to $0.491$.

Panel A shows the same distinction. Firm variance is $0.136$ under every common basis and $0.137$ under every group-specific basis, while sorting stays between $0.075$ and $0.078$. Worker variance moves more, ranging from $0.542$ to $0.551$ under the common bases and from $0.549$ to $0.562$ under the group-specific bases.

\subsection{Interpreting the role of flexible controls}
The within-panel ladder shows that adding observed characteristics materially changes the Panel B decomposition, while alternative functional forms produce smaller changes in its firm and sorting components.
The bulk of the movement comes from including the firm inputs at all, not from how flexibly they enter: once they are in the model linearly, a degree-5 basis reallocates little.
The basis comparison reinforces this conclusion for firm variance and sorting, but not for worker variance: the group-specific B-spline returns the worker component to the linear-additive level. For these covariates in Portuguese data, the standard linear-in-logs adjustment is therefore close to adequate for the firm and sorting components; flexible controls mainly reallocate the worker component.

The small changes in firm variance and sorting do not arise because the rich control design is mechanically unable to detect nonlinear wage schedules.
Relative to heterogeneous linear controls, the fully interacted specification lowers mean squared error by $1.25\%$ in sample and $1.08\%$ out of sample, on a holdout fold drawn so that the training data preserve both regressor support and the connectedness of the worker--firm network.
In a separate signal-recovery exercise, we inject a nonlinear component with variance equal to $5\%$ of the restricted residual variance; the recovered incremental fitted component has a slope of $0.998$ against the injected signal.
The underlying fixed effects also move even when their aggregate second moments change little: worker effects covering $54.5\%$ of observations and firm effects covering $24.0\%$ move to a different effect decile.
Appendix Table \ref{tab:empirical_covariate_validation} reports these diagnostics.
They indicate that incremental nonlinear wage variation is modest in the data and that coefficient reallocation need not translate one-for-one into firm variance or sorting.

The full specification allows the return to age and the transformed firm inputs to differ across demographic groups (gender in Panel A; gender $\times$ education $\times$ qualification in Panel B).

Several caveats apply.
First, this conclusion speaks only to heterogeneity along observed $z_{it}$. Match-specific premia or complementarities outside $z_{it}$ can still matter (e.g., a manager--worker fit that raises pay on a particular team).
Second, because the interaction groups are defined by worker categories rather than firm identifiers, the model does not estimate firm-specific pay schedules \emph{per se}. Firms differ in their effective covariate adjustment only to the extent that they employ different mixes of worker types and have different levels of the continuous firm controls. This is a different margin from the one studied by \textcite{cardcardosokline16}, who let the firm pay premium itself differ by worker group and find in Portuguese data that women receive about ninety percent of the premiums men receive; we let the \emph{control function} vary by group while holding the firm effect common. Allowing firm-level interaction groups, or group-specific firm effects, would require a separate design.
Third, our decomposition is static: time-varying firm policies or match dynamics can exist without affecting the estimated variance components.

\section{Conclusions}\label{sec:conclusions}

AKM decompositions attach economic meaning to worker effects, firm effects, and sorting only after wages have been adjusted for observed characteristics.
Generalized AKM estimates the original variance components while replacing the standard linear adjustment with an unknown smooth function that can interact with worker groups.
The method therefore changes the treatment of observables without changing the decomposition that researchers want to interpret.

We establish when a semiparametric leave-out correction remains valid with heteroskedastic errors, many fixed effects, and a growing series basis.
Because the targets are quadratic, approximation error enters both the plug-in decomposition and its bias correction.
Controlling both terms requires more smoothness than estimation of a linear functional, but the condition is mild for the low-dimensional covariate functions common in wage applications.
The fixed-rank and growing-rank limits cover the main variance-component cases, and the random-projection approximation makes the estimator feasible on large matched datasets.
The simulations show that this flexibility removes bias under strong nonlinearity and costs little when the linear specification is adequate.

The Portuguese application separates two specification choices that are often combined.
Adding education, qualification, and transformed firm inputs materially lowers all three variance components.
Once those controls are present, nonlinearities and worker-group interactions produce smaller additional changes in firm variance and sorting: across the group-specific bases, firm variance remains between $0.114$ and $0.117$ and sorting between $0.041$ and $0.042$.
Worker variance is more sensitive to the basis, ranging from $0.474$ to $0.491$.
The parsimonious age-and-gender decomposition is stable throughout the specification ladder.

The scope of these conclusions is precise.
The theory covers a fixed number of continuous covariates and observation-level leave-out inference.
The application instead uses a cluster-fold point correction, so its cross-specification differences are descriptive and cluster-level inference remains open.
The interactions vary across worker categories rather than firms, and the model does not estimate firm-specific pay schedules.
Extensions to a growing covariate dimension or machine-learning nuisance functions require new arguments.

The empirical lesson is not that nonlinear controls overturn AKM.
In this application, the main specification choice for measured firm variance and sorting is which observed characteristics are removed, not whether included characteristics enter linearly.
Generalized AKM lets the data determine whether linearity is harmless or reallocates wage dispersion, while preserving the worker, firm, and sorting components at the center of the analysis.

\newpage
\printbibliography

\newpage
\begin{appendix}
\makeatletter
\@addtoreset{equation}{section}
\@addtoreset{table}{section}
\makeatother
\renewcommand{\theequation}{\thesection\arabic{equation}}
\renewcommand{\thetable}{\thesection\arabic{table}}
\renewcommand{\theHequation}{appendix.\thesection.\arabic{equation}}
\renewcommand{\theHtable}{appendix.\thesection.\arabic{table}}
\section{Proofs}
This appendix collects the proofs. We first record the leave-one-out residual-variance identity used throughout, then prove the consistency and limiting-distribution results in the order they appear in the paper: Lemma \ref{lem1}, Lemma \ref{jla_lem}, Theorem \ref{chi_sq}, and Theorem \ref{big_r_norm}.

\paragraph*{Leave-one-out residual variance}

The representation in (\ref{sigma_elt}) holds by applying the Sherman--Morrison formula to the leave-one-out estimator of $\gamma$. Denoting $S_{ww} := \sum_{i = 1} ^n w_i w_i'$, it holds that
\begin{equation*}
    \left(S_{ww} - w_i w_i'\right)^{-1} = S_{ww}^{-1} + \frac{S_{ww}^{-1} w_i w_i' S_{ww}^{-1}}{1 - w_i' S_{ww}^{-1} w_i}.
\end{equation*}
Plugging the above into the definition of $\hat{\gamma}_{-i}$, we obtain
\begin{equation*}
\begin{split}
    \hat{\gamma}_{-i} &= \left(S_{ww}^{-1} + \frac{S_{ww}^{-1} w_i w_i' S_{ww}^{-1}}{1 - w_i' S_{ww}^{-1} w_i}\right) \left( \sum_{i = 1} ^n w_i y_i - w_i y_i \right) \\
    &= S_{ww}^{-1} \sum_{i = 1} ^n w_i y_i - S_{ww} w_i y_i + \frac{S_{ww}^{-1} w_i w_i' S_{ww}^{-1} \sum_{i = 1} ^n w_i y_i}{1 - w_i' S_{ww}^{-1} w_i} - \frac{S_{ww}^{-1} w_i w_i' S_{ww}^{-1} w_i y_i}{1 - w_i' S_{ww}^{-1} w_i} \\
    &= \hat{\gamma} + S_{ww}^{-1} w_i \left( \frac{w_i' \hat{\gamma}}{1 - w_i' S_{ww}^{-1} w_i} - y_i - \frac{w_i' S_{ww}^{-1} w_i y_i}{1 - w_i' S_{ww}^{-1} w_i} \right) \\
    &= \hat{\gamma} + S_{ww}^{-1} w_i \left(\frac{w_i' \hat{\gamma} - y_i}{1 - w_i' S_{ww}^{-1} w_i}\right) \\
    &= \hat{\gamma} - S_{ww}^{-1} w_i \frac{(y_i - w_i' \hat{\gamma})}{1 - w_i' S_{ww}^{-1} w_i}.
\end{split}
\end{equation*}
Substituting for $\hat{\gamma}_{-i}$ in (\ref{sigma_hat}), we have
\begin{equation*}
    \begin{split}
        \hat{\sigma}_i ^2 &= y_i \left(y_i - w_i'\hat{\gamma}_{-i}\right) \\
        &= y_i \left(y_i - w_i' \left( \hat{\gamma} - S_{ww}^{-1} w_i \frac{(y_i - w_i' \hat{\gamma})}{1 - w_i' S_{ww}^{-1} w_i} \right)\right) \\
        &= y_i \left(y_i - w_i' \hat{\gamma} + w_i' S_{ww}^{-1} w_i \frac{(y_i - w_i' \hat{\gamma})}{1 - w_i' S_{ww}^{-1} w_i} \right) \\
        &= y_i \left( (y_i - w_i' \hat{\gamma}) \left( \frac{1}{1 - w_i' S_{ww}^{-1} w_i} \right) \right) \\
        &= \frac{y_i \hat{e}_i}{M_{W, ii}}.
    \end{split}
\end{equation*}

\paragraph*{Proof of Lemma \ref{lem1}} 

The variance of $\hat{\beta}$ is
\begin{equation*}
    \begin{split}
        \var[\hat{\beta}] &= \var\left[ S_{xx}^{-1} \sum_{i = 1} ^n \sum_{j = 1} ^n M_{ij} x_i f(z_j) + S_{xx}^{-1} \sum_{i = 1} ^n \sum_{j = 1} ^n M_{ij} x_i e_j \right] \\
        &= S_{xx}^{-1} \var \left[\sum_{i = 1} ^n \sum_{j = 1} ^n M_{ij} x_i f(z_j) + \sum_{i = 1} ^n \sum_{j = 1} ^n M_{ij} x_i e_j \right] S_{xx} ^{-1} \\
        &= S_{xx} ^{-1} \left( \sum_{i = 1} ^n \sum_{j = 1} ^n M_{ij}^2 x_i x_i' \sigma ^2_j \right) S_{xx} ^{-1}
    \end{split}
\end{equation*}
under the assumption of fixed $x_i$ and $z_i$.

First, rewrite the difference between the estimator and the estimand as 
\begin{equation*}
    \begin{split}
        \hat{\theta} - \theta &= \hat{\beta}' A \hat{\beta} - \beta' A \beta - \sum_{i = 1} ^n B_{ii} \hat{\sigma}_i ^2 \\
        &= \sum_{i = 1} ^n \sum_{j = 1} ^n M_{ij} y_j x_i'  S_{xx}^{-1} A S_{xx}^{-1} \sum_{a = 1}^n \sum_{b = 1} ^n M_{ab} x_a y_b - \beta' A \beta - \sum_{i = 1} ^n B_{ii} \hat{\sigma}_i ^2.
    \end{split}
\end{equation*}
Given a data-generating process for $y_i$, we can expand $\hat{\theta} - \theta$ further as
\begin{equation*}
    \begin{split}
        \hat{\theta} - \theta &= \sum_{i = 1} ^n \sum_{j = 1} ^n M_{ij} \left(x_j' \beta + f(z_j) + e_j\right) x_i'  S_{xx}^{-1} A S_{xx}^{-1} \sum_{a = 1}^n \sum_{b = 1} ^n M_{ab} x_a \left(x_b' \beta + f(z_b) + e_b\right) - \beta' A \beta - \sum_{i = 1} ^n B_{ii} \hat{\sigma}_i ^2 \\
        &= \sum_{i = 1} ^n \sum_{j = 1} ^n M_{ij} x_j' \beta x_i' S_{xx}^{-1} A S_{xx}^{-1} \sum_{a = 1} ^n \sum_{b = 1} ^n M_{ab} x_a x_b' \beta + \sum_{i = 1} ^n \sum_{j = 1} ^n M_{ij} f(z_j) x_i' S_{xx}^{-1} A S_{xx}^{-1} \sum_{a = 1} ^n \sum_{b = 1} ^n M_{ab} x_a f(z_b) \\
        &+ \sum_{i = 1} ^n \sum_{j = 1} ^n M_{ij} e_j x_i' S_{xx}^{-1} A S_{xx}^{-1} \sum_{a = 1} ^n \sum_{b = 1} ^n M_{ab} x_a e_b + \sum_{i = 1} ^n \sum_{j = 1} ^n M_{ij} x_j' \beta x_i' S_{xx}^{-1} A S_{xx}^{-1} \sum_{a = 1} ^n \sum_{b = 1} ^n M_{ab} x_a f(z_b) \\
        &+ \sum_{i = 1} ^n \sum_{j = 1} ^n M_{ij} f(z_j) x_i' S_{xx}^{-1} A S_{xx}^{-1} \sum_{a = 1} ^n \sum_{b = 1} ^n M_{ab} x_a x_b' \beta + \sum_{i = 1} ^n \sum_{j = 1} ^n M_{ij} x_j' \beta x_i' S_{xx}^{-1} A S_{xx}^{-1} \sum_{a = 1} ^n \sum_{b = 1} ^n M_{ab} x_a e_b \\
        &+ \sum_{i = 1} ^n \sum_{j = 1} ^n M_{ij} e_j x_i' S_{xx}^{-1} A S_{xx}^{-1} \sum_{a = 1} ^n \sum_{b = 1} ^n M_{ab} x_a x_b' \beta + \sum_{i = 1} ^n \sum_{j = 1} ^n M_{ij} f(z_j) x_i' S_{xx}^{-1} A S_{xx}^{-1} \sum_{a = 1} ^n \sum_{b = 1} ^n M_{ab} x_a e_b \\
        &+ \sum_{i = 1} ^n \sum_{j = 1} ^n M_{ij} e_j x_i' S_{xx}^{-1} A S_{xx}^{-1} \sum_{a = 1} ^n \sum_{b = 1} ^n M_{ab} x_a f(z_b) - \beta' A \beta - \sum_{i = 1} ^n B_{ii} \hat{\sigma}_i ^2.
    \end{split}
\end{equation*}
Next, we use the definition $B_{i\ell} := \sum_{j = 1} ^n M_{ij} x_j' S_{xx} ^{-1} A S_{xx} ^{-1} \sum_{j = 1} ^n M_{\ell j} x_j$, so that
\begin{equation*}
    \begin{split}
        \hat{\theta} - \theta &= \sum_{i = 1} ^n \sum_{\ell = 1} ^n B_{i \ell} f(z_i)^2 + \sum_{i = 1} ^n \sum_{\ell = 1} ^n B_{i \ell} e_i ^2 - \sum_{i = 1} ^n B_{ii} \hat{\sigma}_i ^2 \\
        &+ 2 \sum_{i = 1} ^n \sum_{\ell = 1} ^n B_{i \ell} x_\ell' \beta f(z_i) + 2 \sum_{i = 1} ^n \sum_{\ell = 1} ^n B_{i \ell} x_\ell' \beta e_i + 2 \sum_{i = 1} ^n \sum_{\ell = 1} ^n B_{i \ell} e_\ell f(z_i).
    \end{split}
\end{equation*}
Finally, rearranging, we have
\begin{equation}
    \begin{split}\label{theta_diff}
        \hat{\theta} - \theta &= \sum_{i = 1} ^n \sum_{\ell = 1} ^n B_{i \ell} f(z_i)^2 + \sum_{i = 1} ^n \sum_{\ell \neq i}  B_{i \ell} e_i e_\ell + \sum_{i = 1} ^n B_{ii} (e_i^2 - \hat{\sigma}_i ^2) \\
        &+ 2 \sum_{i = 1} ^n \sum_{\ell = 1} ^n B_{i \ell} x_\ell' \beta f(z_i) + 2 \sum_{i = 1} ^n \sum_{\ell = 1} ^n B_{i \ell} x_\ell' \beta e_i + 2 \sum_{i = 1} ^n \sum_{\ell = 1} ^n B_{i \ell} e_\ell f(z_i).
    \end{split}
\end{equation}
In general, it holds that
\begin{equation*}
    \E[|\hat{\theta} - \theta|^2] = |\E[\hat{\theta} - \theta]|^2 + \text{trace}(\var[\hat{\theta} - \theta]).
\end{equation*}
To show that $\hat{\theta}$ is consistent for $\theta$, we need to show that the bias and the variance of the difference in (\ref{theta_diff}) goes to zero. The main idea is to compute bounds on each term and their variances, and show that these bounds are asymptotically negligible. Then, convergence in the quadratic mean would imply convergence in probability.

Applying expectations on both sides, and using independence and mean-zero properties of errors, we have
\begin{equation*}
    \E[\hat{\theta} - \theta] = \sum_{i = 1} ^n \sum_{\ell = 1} ^n B_{i \ell} f(z_i)^2 + 2 \sum_{i = 1} ^n \sum_{\ell = 1} ^n B_{i \ell} x_\ell' \beta f(z_i) + \sum_{i = 1} ^n B_{ii} \E[\sigma_i ^2 - \hat{\sigma}_i ^2].
\end{equation*}
We now prove that each term goes to zero in probability as $n \rightarrow \infty$, $k \rightarrow \infty$, and $p \rightarrow \infty$.

Denote $B := (B_{i\ell})_{i, \ell = 1} ^n \in \mathbb{R}^{n \times n}$, and $F := (f(z_1), \ldots, f(z_n))' \in \mathbb{R}^n$. Then
\begin{equation*}
    \begin{split}
        \sum_{i = 1} ^n \sum_{\ell = 1} ^n B_{i \ell} f(z_i)^2 &= F'BF \\
        &= F'MX S_{xx}^{-1} A S_{xx}^{-1} X'MF \\
        &= F'MX S_{xx}^{-1} A^{1/2} A^{1/2} S_{xx}^{-1} X'MF \\
        &= (F'MX S_{xx}^{-1} A^{1/2})^2,
    \end{split}
\end{equation*}
where we use $A = A^{1/2} A^{1/2}$ because $A$ is symmetric. By the Markov inequality, assumptions on $\mathcal{F}$, $M$ being idempotent, and the Cauchy-Schwarz inequality,
\begin{equation*}
    \begin{split}
        \norm{\frac{1}{n} F'MX S_{xx}^{-1} A^{1/2}} &\leq \text{trace} \left(\frac{1}{n} F'MF \right)^{1/2} \cdot \text{trace} \left(\frac{1}{n} A^{1/2} S_{xx}^{-1} X'M X S_{xx}^{-1} A^{1/2} \right)^{1/2} \\
        &= \text{trace} \left(\frac{1}{n} F'MF \right)^{1/2} \cdot \text{trace} \left(\frac{1}{n} A^{1/2} S_{xx}^{-1} A^{1/2} \right)^{1/2} \\
        &= \mathcal{O} (k^{- \alpha_f} \sqrt{p / n}).
    \end{split}
\end{equation*}
Thus, using the Assumption \ref{asum1} and \ref{asum3}, we have 
\begin{equation*}
    \sum_{i = 1} ^n \sum_{\ell = 1} ^n B_{i \ell} f(z_i)^2 = \mathcal{O} (k^{-2 \alpha_f} np) = \mathcal{O} (k^{-\alpha_f} n) \rightarrow 0.
\end{equation*}

Now, for the second term in the main decomposition,
\begin{equation*}
    \begin{split}
        \sum_{i = 1} ^n \sum_{\ell = 1} ^n B_{i \ell} x_\ell' \beta f(z_i) &= F'BX\beta \\
        &= F'MXS_{xx}^{-1}A S_{xx}^{-1} X'M X \beta \\
        &= F'MXS_{xx}^{-1} A \beta.
    \end{split}
\end{equation*}
Given that $\beta' A S_{xx}^{-1} A \beta = \mathcal{O} (1)$, by the Markov inequality, assumption on $\mathcal{F}$, $M$ being idempotent, and the Cauchy-Schwarz inequality we have
\begin{equation*}
    \begin{split}
        \norm{\frac{1}{n} F'MXS_{xx}^{-1} A \beta} &\leq \text{trace} \left(\frac{1}{n} F'MF \right)^{1/2} \cdot \text{trace} \left(\frac{1}{n} \beta' A S_{xx}^{-1} X'MX S_{xx}^{-1} A \beta \right)^{1/2} \\
        &= \text{trace} \left(\frac{1}{n} F'MF \right)^{1/2} \cdot \text{trace} \left(\frac{1}{n} \beta' A S_{xx}^{-1} A \beta \right)^{1/2} \\
        &= \mathcal{O} (k^{-\alpha_f} / \sqrt{n}).
    \end{split}
\end{equation*}
So that using the Assumption \ref{asum1} and \ref{asum3}, we have
\begin{equation*}
    \sum_{i = 1} ^n \sum_{\ell = 1} ^n B_{i \ell} x_\ell' \beta f(z_i) = \mathcal{O} (k^{-\alpha_f} \sqrt{n}) \rightarrow 0.
\end{equation*}

To bound the third term, we should show that
\begin{equation*}
    \E[\hat{\sigma}_i ^2 -\sigma_i ^2] = \left(x_i'\beta + f(z_i)\right) \left(x_i' \E[\beta - \hat{\beta}_{-i}] + \E[f(z_i) - p_k(z_i)' \hat{\alpha}_{-i}]\right)
\end{equation*}
goes in probability to zero. This is equivalent to bounding the bias of the estimator of the slope coefficient and the bias of the function approximation. 

The bias of the slope estimator is
\begin{equation*}
    \E[\hat{\beta} - \beta] = \E\left[\left(\frac{1}{n} X'MX\right)^{-1} \frac{1}{n} X'MF\right].
\end{equation*}
By the Markov inequality, assumption on $\mathcal{F}$, $M$ being idempotent, and the Cauchy-Schwarz inequality we have that:
\begin{equation*}
    \begin{split}
        \norm{\frac{1}{n} X'MF} &\leq \text{trace}\left(\frac{1}{n} X'MX\right)^{1/2} \cdot \text{trace}\left(\frac{1}{n} F'MF\right)^{1/2} \\
        &= \mathcal{O}(k ^{-\alpha_f} \sqrt{p / n}),
    \end{split}
\end{equation*}
and, similarly,
\begin{equation*}
    \norm{\frac{1}{n} X'MX} = \mathcal{O}(p / n),
\end{equation*}
so that using Assumption \ref{asum1} and \ref{asum3}, we ultimately have:
\begin{equation*}
    \E[\norm{\hat{\beta} - \beta}] = \mathcal{O}(\sqrt{n} / k^{\alpha_f} \sqrt{p}) \rightarrow 0.
\end{equation*}

Because of the assumption on the functional class $\mathcal{F}$, we can bound the bias of the function approximation using Jensen's inequality for convex $x \mapsto x^2$,
\begin{equation*}
    \begin{split}
        \E[|f(z_i) - p_k (z_i)' \hat{\alpha}_{-i}|] &\leq \E[|f(z_i) - p_k (z_i)' \hat{\alpha}_{-i}|^2]^{1/2} \\
        &\leq (Ck^{-2\alpha_f})^{1/2} = \mathcal{O}_p(k^{-\alpha_f}),
    \end{split}
\end{equation*}
so that:
\begin{equation*}
    \E[\lvert\hat{\sigma}_i ^2 -\sigma_i ^2\rvert] = \mathcal{O}_p(\sqrt{n} / k^{\alpha_f} \sqrt{p}) \overset{p}{\rightarrow} 0.
\end{equation*}
From this it follows that:
\begin{equation*}
    \E[\lvert\hat{\theta} - \theta\rvert] \overset{p}{\rightarrow} 0.
\end{equation*}

We next turn our attention to bounding the variance. Let matrix $$\tilde{A} := S_{xx}^{-1/2} A S_{xx}^{-1/2},$$
and let $\lambda_1, \ldots, \lambda_r$ be its nonzero eigenvalues. We assume that $\lambda_1 ^2 \geq \ldots \geq \lambda_r^2$, and that each eigenvalue appears as many times as its algebraic multiplicity. Under these assumptions, we can spectrally decompose $\tilde{A} = Q D Q'$, where $Q$ is a matrix of orthonormal vectors and $D = \text{diag} (\lambda_1, \ldots, \lambda_r)$. 

The variance of $2 \sum_{i = 1} ^n \sum_{\ell = 1} ^n B_{i \ell} x_\ell' \beta e_i$ is
\begin{equation*}
    \begin{split}
        4 \sum_{i = 1} ^n \left( \sum_{\ell = 1} ^n B_{i\ell} x_\ell ' \beta \right)^2 \sigma_i ^2 &\leq \max_i \sigma_i ^2 \beta' X'B^2 X\beta = \max_i \sigma_i ^2 \beta' A S_{xx}^{-1} A \beta \\
        &\leq \max_i \sigma_i ^2 \lambda_1 \theta = o(1).
    \end{split}
\end{equation*}
To explain why the last inequality holds, define $\tilde{\beta} = S_{xx}^{1/2} \beta$ so that $\beta = S_{xx}^{-1/2} \tilde{\beta}$,
\begin{equation*}
    \beta' A S_{xx}^{-1} A \beta = \tilde{\beta}' S_{xx}^{-1/2} A S_{xx}^{-1/2} \tilde{\beta} = \tilde{\beta}' (S_{xx}^{-1/2} A S_{xx}^{-1/2})^2 \tilde{\beta} = \tilde{\beta}' \tilde{A}^2 \tilde{\beta},
\end{equation*}
and
\begin{equation*}
    \theta = \tilde{\beta}' S_{xx}^{-1/2} A S_{xx}^{-1/2} \tilde{\beta} = \tilde{\beta}' \tilde{A} \tilde{\beta}.
\end{equation*}
Then, using the Rayleigh quotient argument, it holds that
\begin{equation*}
    \frac{\tilde{\beta}' \tilde{A}^2 \tilde{\beta}}{\tilde{\beta}' \tilde{A} \tilde{\beta}} \leq \lambda_{\text{max}} (\tilde{A}),
\end{equation*}
because $\tilde{A}$ is positive semi-definite, $\tilde{\beta} \neq 0$, and $\lambda_{\text{max}} (\tilde{A})$ is the largest eigenvalue of $\tilde{A}$. From this it follows that we can bound
\begin{equation*}
    \beta' A S_{xx}^{-1} A \beta \leq \lambda_1 \theta,
\end{equation*}
use assumptions $\theta = \mathcal{O}(1)$, $\lambda_1 \leq \text{trace}(\tilde{A^2})^{1/2} = o(1)$, and the fact that the variance is bounded. The conclusion then follows. 

The variance of $2 \sum_{i = 1} ^n \sum_{\ell = 1} ^n B_{i \ell} e_\ell f(z_i)$ is
\begin{equation*}
    \begin{split}
        4 \sum_{i = 1} ^n \sum_{\ell = 1} ^n B_{i\ell}^2 \sigma^2 _\ell f(z_i)^2 &\leq \max_i 4 \sigma_i ^2 f(z_i) ^2 \sum_{i = 1} ^n \sum_{\ell = 1} ^n B_{i\ell}^2 = \max_i 4 \sigma_i ^2 f(z_i) ^2 \text{trace} (\tilde{A}^2) = o(1),
    \end{split}
\end{equation*}
which follows since
\begin{equation*}
    \begin{split}
        \sum_{i = 1} ^n \sum_{\ell = 1} ^n B_{i \ell} ^2 &= \norm{B}_F ^2 = \text{trace}(B'B) \\
        &= \text{trace}(MX S_{xx}^{-1} A S_{xx}^{-1} A S_{xx}^{-1} X'M) \\
        &= \text{trace} (MX S_{xx}^{-1/2} S_{xx}^{-1/2} A S_{xx}^{-1/2} S_{xx}^{-1/2} A S_{xx}^{-1/2} S_{xx}^{-1/2} X'M) \\
        &= \text{trace} (S_{xx}^{-1/2} A S_{xx}^{-1/2} S_{xx}^{-1/2} A S_{xx}^{-1/2} S_{xx}^{-1/2} X'M X S_{xx}^{-1/2}) \\
        &= \text{trace} (S_{xx}^{-1/2} A S_{xx}^{-1/2} S_{xx}^{-1/2} A S_{xx}^{-1/2}) = \text{trace} (\tilde{A}^2),
    \end{split}
\end{equation*}
and we use the assumptions as above but now, instead of the bounded variance, we assume that $\mathcal{F}$ is a class of bounded functions so that $\max |f(z_i)| < C$ for $i = 1, \ldots, n$ for some absolute constant $C$ (this is implied by the main assumption on $\mathcal{F}$).

Because $B_{i\ell}^2 = B_{\ell i}^2$ for any $i, \ell$, the variance of $\sum_{i = 1} ^n \sum_{\ell \neq i}  B_{i \ell} e_i e_\ell$ is
\begin{equation*}
    2 \sum_{i = 1} ^n \sum_{\ell \neq i}  B_{i \ell}^2 \sigma^2_i \sigma^2_\ell \leq \max_i 2 \sigma_i ^4 \sum_{i = 1}^n \sum_{\ell = 1} ^n B_{i\ell}^2 = \max_i 2 \sigma_i ^4 \text{trace} (\tilde{A}^2) = o(1).
\end{equation*}

To compute the variance of $\sum_{i = 1} ^n B_{ii} (e_i^2 - \hat{\sigma}_i ^2)$, we represent the leave-one-out variance estimator as
\begin{equation*}
    \begin{split}
        \hat{\sigma}_i ^2 &= y_i (y_i - w_i' \hat{\gamma}_{-i}) = y_i M_{W, ii}^{-1} (y_i - w_i' \hat{\gamma}) \\
        &= y_i M_{W, ii}^{-1} \hat{e}_i = y_i M_{W, ii}^{-1} \sum_{\ell = 1} ^n M_{W, i\ell} y_\ell \\
        &= y_i M_{W, ii}^{-1} \sum_{\ell = 1} ^n M_{W, i\ell} \left(x_\ell ' \beta + f(z_\ell) + e_\ell\right) \\
        &= y_i M_{W, ii}^{-1} \sum_{\ell = 1} ^n M_{W, i\ell} \left(x_\ell' \beta + f(z_\ell)\right) + y_i M_{W, ii}^{-1} \sum_{\ell = 1} ^n M_{W, i\ell} e_\ell.
    \end{split}
\end{equation*}  
Thus, the variance is now expressed as:
\begin{multline*}
        \sum_{i = 1} ^n \left( \sum_{\ell = 1} ^n M_{W, \ell \ell} ^{-1} B_{\ell \ell} M_{W, i\ell} \left(x_\ell '\beta + f(z_\ell)\right) \right)^2 \sigma_i ^2 + 2 \sum_{i = 1} ^n \sum_{\ell \neq i} M_{W, ii}^{-2} B_{ii}^2 M_{W, i\ell}^2 \sigma_i ^2 \sigma_\ell ^2 \\ 
        \leq \frac{1}{c^2} \max_i \sigma_i ^2 \max_i \left(x_i' \beta + f(z_i)\right)^2 \sum_{i = 1} ^n B_{ii}^2 + \frac{2}{c} \max_i \sigma_i ^4 \sum_{i = 1} ^n B_{ii}^2 = o(1),
\end{multline*}
because $\min_i M_{W, ii} \geq c > 0$, $\sum_{i = 1} ^n B_{ii} ^2 \leq \text{trace}(\tilde{A}^2) = o(1)$, and $\max_i \left(x_i' \beta + f(z_i)\right)^2 \leq 2 \max_i (x_i' \beta)^2 + 2 \max_i f(z_i)^2 = \mathcal{O}(1)$. 

Because we have that
\begin{equation*}
    \E[\hat{\theta} - \theta] \overset{p}{\rightarrow} 0, \quad \var[\hat{\theta} - \theta] \overset{p}{\rightarrow} 0,
\end{equation*}
the proposed estimator $\hat{\theta}$ is consistent. \hfill \qedsymbol

\paragraph*{Proof of Lemma \ref{jla_lem}} We prove the result by considering a second-order approximation of $\hat{\theta}_{\text{JLA}} - \hat{\theta}$ around $\hat{a}_i := (1 - P_{W, ii})^{-1} (\hat{P}_{W, ii} - P_{W, ii})$ as
    \begin{equation*}
        (\hat{\theta}_{\text{JLA}} - \hat{\theta})_2 := \sum_{i = 1} ^n \hat{\sigma}_i ^2 \left(B_{ii} - \hat{B}_{ii} - \hat{B}_{ii} \hat{a}_i - \hat{B}_{ii} \left(\hat{a}_i ^2 - \frac{1}{m} \frac{3 P_{W, ii}^3 + P_{W, ii}^2}{1 - P_{W, ii}}\right) \right),
    \end{equation*}
    and an approximation error that we show to be negligible,
    \begin{equation*}
        \text{AE}_2 := \sum_{i = 1} ^n \hat{\sigma}_i ^2 \hat{B}_{ii} \left( \frac{1}{m} \frac{3 \hat{P}_{W, ii} ^3 + \hat{P}_{W, ii} ^2 - (3 P_{W, ii}^3 + P_{W, ii}^2) (1 - \hat{a}_i)^2}{(1 - \hat{a}_i)^2 (1 - P_{W, ii})} - \frac{\hat{a}_i ^3}{1 - \hat{a}_i} \right).
    \end{equation*}
    To decompose $\hat{\theta}_{\text{JLA}} - \hat{\theta} = (\hat{\theta}_{\text{JLA}} - \hat{\theta})_2 + \text{AE}_2$, note that
    \begin{equation*}
        \begin{split}
            \hat{\sigma}^2_{i, \text{JLA}} &= \frac{y_i (y_i - w_i' \hat{\gamma})}{1 - \hat{P}_{W, ii}} \left(1 - \frac{1}{m} \frac{3 \hat{P}_{W, ii} ^3 + \hat{P}_{W, ii} ^2}{1 - \hat{P}_{W, ii}}\right) \\
            &= \frac{1 - P_{W, ii}}{1 - \hat{P}_{W, ii}} \hat{\sigma}_i ^2 \left(1 - \frac{1}{m} \frac{3 \hat{P}_{W, ii} ^3 + \hat{P}_{W, ii} ^2}{1 - \hat{P}_{W, ii}}\right),
        \end{split}
    \end{equation*}
    so that
    \begin{equation*}
        \begin{split}
            \hat{\theta}_{\text{JLA}} - \hat{\theta} &= \sum_{i = 1} ^n B_{ii} \hat{\sigma}_i ^2 - \hat{B}_{ii} \hat{\sigma}^2_{i, \text{JLA}} \\
            &= \sum_{i = 1} ^n B_{ii} \hat{\sigma}_i ^2 - \hat{B}_{ii} \hat{\sigma}_i ^2 \left(1 - \frac{1}{m} \frac{3 \hat{P}_{W, ii} ^3 + \hat{P}_{W, ii} ^2}{1 - \hat{P}_{W, ii}}\right) \\
            &= \sum_{i = 1} ^n B_{ii} \hat{\sigma}_i ^2 + \hat{B}_{ii} \hat{\sigma}_i ^2 \left(\frac{1}{m} \frac{1 - P_{W, ii}}{1 - \hat{P}_{W, ii}} \frac{3 \hat{P}_{W, ii}^3 + \hat{P}_{W, ii}^2}{1 - \hat{P}_{W, ii}} - \frac{1 - P_{W, ii}}{1 - \hat{P}_{W, ii}}\right).
        \end{split}
    \end{equation*}
    Add and subtract $m^{-1} \hat{\sigma}_i ^2 \hat{B}_{ii} (1 - P_{W, ii})^{-1} (3 P_{W, ii}^3 + P_{W, ii}^2)$ to obtain
    \begin{equation*}
        \begin{split}
            \hat{\theta}_{\text{JLA}} - \hat{\theta} &= \sum_{i = 1} ^n B_{ii} \hat{\sigma}_i ^2 + \frac{1}{m} \hat{\sigma}_i ^2 \hat{B}_{ii} \frac{3 P_{W, ii}^3 + P_{W, ii}^2}{1 - P_{W, ii}} + \hat{B}_{ii} \hat{\sigma}_i ^2 \left(\frac{1}{m} \frac{1 - P_{W, ii}}{1 - \hat{P}_{W, ii}} \frac{3 \hat{P}_{W, ii}^3 + \hat{P}_{W, ii}^2}{1 - \hat{P}_{W, ii}} - \frac{1 - P_{W, ii}}{1 - \hat{P}_{W, ii}}\right) \\&
            - \frac{1}{m} \hat{\sigma}_i ^2 \hat{B}_{ii} \frac{3 P_{W, ii}^3 + P_{W, ii}^2}{1 - P_{W, ii}} \\
            &= \sum_{i = 1} ^n \hat{\sigma}_i ^2 \left(B_{ii} + \frac{1}{m} \hat{B}_{ii} \frac{3 P_{W, ii}^3 + P_{W, ii}^2}{1 - P_{W, ii}}\right) \\
            &+ \sum_{i = 1} ^n \hat{\sigma}_i ^2 \hat{B}_{ii} \left(\frac{1}{m} \frac{(1 - P_{W, ii})^2 (3 \hat{P}_{W, ii}^3 + \hat{P}_{W, ii} ^2) - (1 - \hat{P}_{W, ii})^2 (3 P_{W, ii}^3 + P_{W, ii}^2)}{(1 - \hat{P}_{W, ii})^2 (1 - P_{W, ii})} - \frac{1 - P_{W, ii}}{1 - \hat{P}_{W, ii}}\right),
        \end{split}
    \end{equation*}
    and using $(1 - \hat{P}_{W, ii})^{-1} (1 - P_{W, ii}) = (1 - \hat{a}_i)^{-1}$, and expanding up to the third order as
    \begin{equation*}
        \begin{split}
            \frac{1 - P_{W, ii}}{1 - \hat{P}_{W, ii}} &= 1 + \frac{1 - P_{W, ii}}{1 - \hat{P}_{W, ii}} \hat{a}_i \\
            &= 1 + \hat{a}_i + \hat{a}_i ^2 + \frac{1 - P_{W, ii}}{1 - \hat{P}_{W, ii}} \hat{a}_i ^3,
        \end{split}
    \end{equation*}
    we have that
    \begin{equation*}
        \begin{split}
            \hat{\theta}_{\text{JLA}} - \hat{\theta} &= \sum_{i = 1} ^n \hat{\sigma}_i ^2 \left(B_{ii} + \frac{1}{m} \hat{B}_{ii} \frac{3 P_{W, ii}^3 + P_{W, ii}^2}{1 - P_{W, ii}} \right) \\
            &+ \sum_{i = 1} ^n \hat{\sigma}_i ^2 \hat{B}_{ii} \left( \frac{1}{m} \frac{3 \hat{P}_{W, ii} ^3 + \hat{P}_{W, ii}^2 - (1 - \hat{a}_i)^2 (3 P_{W, ii}^3 + P_{W, ii}^2)}{(1 - \hat{a}_i)^2 (1 - P_{W, ii})} - \left(1 + \hat{a}_i + \hat{a}_i ^2 + \frac{\hat{a}_i ^3}{1 - \hat{a}_i}\right) \right) \\
            &= (\hat{\theta}_{\text{JLA}} - \hat{\theta})_2 + \text{AE}_2.
        \end{split}
    \end{equation*}

To describe the bias, we note that $\hat{P}_{W, ii}$, $\hat{B}_{ii}$, and $\hat{\sigma}_i ^2$ are independent of each other, $\E[\hat{P}_{W, ii}] = P_{W, ii}$, $\E[\hat{B}_{ii}] = B_{ii}$, $\E[\hat{\sigma}_i ^2] = \sigma_i ^2 + \mathcal{O}_p (\sqrt{n} / k^{\alpha_f} \sqrt{p})$, and using properties of the Rademacher random variables,
\begin{equation*}
    \var[\hat{a}_i] = \frac{2}{m} \frac{P_{W, ii} - \sum_{\ell = 1} ^n P_{W, i \ell}^4}{(1 - P_{W, ii})^2} = \frac{1}{m} \frac{3 P_{W, ii}^3 + P_{W, ii}^2}{1 - P_{W, ii}} + \frac{P_{W, ii} (1 - P_{W, ii})^2 - 2 \sum_{\ell \neq i} ^n P_{W, i\ell}^4}{m (1 - P_{W, ii})^2}.
\end{equation*}
Therefore, in total we have:
\begin{equation*}
    \E[(\hat{\theta}_{\text{JLA}} - \hat{\theta})_2] = - \sum_{i = 1} ^n \sigma_i ^2 B_{ii} \left(\var[\hat{a}_i] - \frac{1}{m} \frac{3 P_{W, ii}^3 + P_{W, ii}^2}{1 - P_{W, ii}}\right) + \mathcal{O}_p (\sqrt{n} / k^{\alpha_f} \sqrt{p}),
\end{equation*}
or, assuming $k \rightarrow \infty$,
\begin{equation*}
    \E[(\hat{\theta}_{\text{JLA}} - \hat{\theta})_2] = \mathrm{B}_m + o(1), \quad \mathrm{B}_m := \sum_{i = 1} ^n B_{ii} \sigma_i ^2 \left(\frac{2 \sum_{\ell \neq i} P_{W, i\ell}^4 - P_{W, ii}^2 (1 - P_{W, ii})^2}{m (1 - P_{W, ii})^2} \right).
\end{equation*}

Focusing on the variance next, denote $y := (y_1, \ldots, y_n)'$, so that
\begin{equation*}
    \begin{split}
        \var\left[\sum_{i = 1} ^n \hat{\sigma}_i ^2 (B_{ii} - \hat{B}_{ii})\right] &= \E\left[\var\left[\sum_{i = 1} ^n \hat{\sigma}_i ^2 \hat{B}_{ii}\right] \middle| y\right] + \var\left[\E\left[\sum_{i = 1} ^n \hat{\sigma}_i ^2 \hat{B}_{ii}\right] \middle| y\right] = \E\left[\var\left[\sum_{i = 1} ^n \hat{\sigma}_i ^2 \hat{B}_{ii}\right] \middle| y\right] \\
        &\leq 2 m^{-1} \sum_{i = 1} ^n \sum_{\ell = 1} ^n B_{i \ell}^2 \E[\hat{\sigma}_i ^2 \hat{\sigma}_\ell ^2] = \mathcal{O}\left(m^{-1} \text{trace} (\tilde{A}^2)\right), 
    \end{split}
\end{equation*}
\begin{equation*}
    \begin{split}
        \var\left[\sum_{i = 1} ^n \hat{\sigma}_i ^2 \hat{B}_{ii} \hat{a}_i\right] &= \E\left[\var\left[ \sum_{i = 1} ^n \hat{\sigma}_i ^2 \hat{B}_{ii} \hat{a}_i\middle | y, R_B\right]\right] \leq 2 m^{-1} \sum_{i = 1} ^n \sum_{\ell = 1} ^n P_{W, i \ell} ^2 \frac{\E[\hat{B}_{ii} \hat{B}_{\ell \ell}] \E[\hat{\sigma}_i ^2 \hat{\sigma}_\ell ^2]}{(1 - P_{W, ii}) (1 - P_{W, \ell \ell})} \\
        &= \mathcal{O}\left(m^{-1} \text{trace} (\tilde{A}^2) + m^{-2} \text{trace} (\tilde{A}_1 ^2)^{1/2} \text{trace} (\tilde{A}_2 ^2)^{1/2}\right)
    \end{split}
\end{equation*}
for $\tilde{A}_k := S_{xx}^{-1/2} A_k' A_k S_{xx} ^{-1/2}$ for $k = 1, 2$. Regarding the ensuing terms, it holds that:
\begin{equation*}
    \begin{split}
        \var\left[\sum_{i = 1} ^n \hat{\sigma}_i ^2 \hat{B}_{ii} \left(\hat{a}_i ^2 - \var[\hat{a}_i]\right)\right] &= \sum_{i = 1} ^n \sum_{\ell = 1} ^n \E[\hat{B}_{ii} \hat{B}_{\ell \ell}] \E[\hat{\sigma}_i ^2 \hat{\sigma}_\ell ^2] \cov[\hat{a}_i ^2, \hat{a}_\ell ^2] \\
        &= \mathcal{O} \left(m^{-2} \text{trace} (\tilde{A}^2) + m^{-3} \text{trace} (\tilde{A}_1 ^2)^{1/2} \text{trace} (\tilde{A}_2 ^2)^{1/2}\right),
    \end{split}
\end{equation*}
\begin{equation*}
    \begin{split}
        \var\left[ \sum_{i = 1} ^n \hat{\sigma}_i ^2 (\hat{B}_{ii} - B_{ii}) \frac{2 \sum_{\ell \neq i} ^n P_{W, i\ell}^4 - P_{W, ii} (1 - P_{W, ii})^2}{m (1 - P_{W, ii})^2} \right] = \mathcal{O}\left(m^{-3} \text{trace} (\tilde{A}^2) \right),
    \end{split}
\end{equation*}
\begin{equation*}
    \var\left[\sum_{i = 1} ^n B_{ii} \left(\hat{\sigma}_i ^2 - \sigma_i ^2\right) \frac{2 \sum_{\ell \neq i} ^n P_{W, i\ell}^4 - P_{W, ii} (1 - P_{W, ii})^2}{m (1 - P_{W, ii})^2} \right] = \mathcal{O} \left(m^{-2} \var[\hat{\theta}]\right).
\end{equation*}
Because $\text{trace} (\tilde{A}^2) = \mathcal{O}(\var[\hat{\theta}])$ and $m^{-4} \var[\hat{\theta}]^{-2} \var[\hat{\theta}_1] \var[\hat{\theta}_2] = o(1)$, it can be ultimately established that $\var[\hat{\theta}]^{-1/2} \left((\hat{\theta}_{\text{JLA}} - \hat{\theta})_2 - \mathrm{B}_m\right) = o_p(1)$.

Using that $\E[\hat{a}_i^3] = \mathcal{O}(m^{-2})$, $\E[\hat{a}_i ^4] = \mathcal{O}(m^{-2})$, and $\max_i |\hat{a}_i| = o_p (\log n / \sqrt{m})$, the terms in the approximation error are as follows:
\begin{gather*}
    \sum_{i = 1} ^n \hat{\sigma}_i ^2 \hat{B}_{ii} \hat{a}_i ^3 + \sum_{i = 1} ^n \hat{\sigma}_i ^2 \hat{B}_{ii} \hat{a}_i ^4 = m^{-2} \mathcal{O}_p \left(\E[\hat{\theta}_{1, \text{PI}} - \theta_1] + \E[\hat{\theta}_{2, \text{PI}} - \theta_2] \right), \\
    \sum_{i = 1} ^n \hat{\sigma}_i ^2 \hat{B}_{ii} \frac{\hat{a}_i ^5}{1 - \hat{a}_i} = \frac{\log n}{m^{5/4}} \mathcal{O}_p \left(\E[\hat{\theta}_{1, \text{PI}} - \theta_1] + \E[\hat{\theta}_{2, \text{PI}} - \theta_2]\right), \\
    \frac{1}{m} \sum_{i = 1} ^n \hat{\sigma}_i ^2 \hat{B}_{ii} \frac{3 \hat{P}_{W, ii}^3 + \hat{P}_{W, ii} ^2 - (3 P_{W, ii}^3 + P_{W, ii}^2) (1 - \hat{a}_i)^2}{(1 - \hat{a}_i)^2 (1 - P_{W, ii})} \\
    = \left(m^{-2} + \frac{\log n}{p^{5/4}}\right) \mathcal{O}_p \left(\E[\hat{\theta}_{1, \text{PI}} - \theta_1] + \E[\hat{\theta}_{2, \text{PI}} - \theta_2]\right).
\end{gather*}
\hfill \qedsymbol

\paragraph*{Proof of Theorem \ref{chi_sq}}

Representation in the theorem holds because
\begin{equation*}
    \sum_{\ell = 1} ^r \lambda_\ell \hat{b}_\ell ^2 = \hat{\beta}' S_{xx}^{1/2} Q D Q' S_{xx}^{1/2} \hat{\beta} = \hat{\beta}' S_{xx}^{1/2} S_{xx}^{-1/2} A S_{xx}^{-1/2} S_{xx}^{1/2} \hat{\beta} = \hat{\beta}' A \hat{\beta},
\end{equation*}
and
\begin{equation*}
    \sum_{i = 1} ^n B_{ii} \hat{\sigma}_i ^2 = \text{trace} (A \widehat{\var}[\hat{\beta}]) = \text{trace} (D \widehat{\var}[\hat{b}]) = \sum_{\ell = 1} ^r \lambda_\ell \widehat{\var}[\hat{b}_\ell].
\end{equation*}
We prove the theorem in three steps.
\\
\textit{Approximation}. Equivalently, represent $\hat{\theta}$ as
\begin{equation*}
    \hat{\theta} = \sum_{\ell = 1} ^r \lambda_\ell \left(\hat{b}_\ell ^2 - \var[\hat{b}_\ell]\right) + \sum_{i = 1} ^n B_{ii} (\sigma_i ^2 - \hat{\sigma}_i ^2),
\end{equation*}
and below we show that the second term is asymptotically dominated by the variance of the estimator, $\var[\hat{\theta}]$. Given that, the asymptotic distribution of $\hat{\theta}$ is then driven by the joint distribution of random vector $\hat{b}$. 

We need to show that the second term is asymptotically mean-zero and is of smaller order than variance of $\hat{\theta}$. The first claim is immediate because we have shown in Lemma \ref{lem1} that
\begin{equation*}
    \E[\hat{\sigma}_i ^2 - \sigma_i ^2] = \mathcal{O}_p(\sqrt{n} / k^{\alpha_f} \sqrt{p}) \overset{p}{\rightarrow} 0,
\end{equation*}
so that
\begin{equation*}
    \E\left[\sum_{i = 1} ^n B_{ii} (\sigma_i ^2 - \hat{\sigma}_i ^2)\right] \overset{p}{\rightarrow} 0.
\end{equation*}
Now, rewrite
\begin{equation*}
    \begin{split}
        \sum_{i = 1} ^n B_{ii} (\hat{\sigma}_i ^2 - \sigma_i ^2) &= \sum_{i = 1} ^n B_{ii} M_{W, ii}^{-1}  x_i' \beta \sum_{\ell = 1} ^n M_{W, i \ell} e_\ell + \sum_{i = 1} ^n (e_i ^2 - \sigma_i ^2) \\
        &+ \sum_{i = 1} ^n B_{ii} M_{W, ii}^{-1}  f(z_i) \sum_{\ell = 1} ^n M_{W, i \ell} e_\ell + \sum_{i = 1} ^n B_{ii} M_{W, ii}^{-1} \sum_{\ell \neq i} M_{W, i\ell} e_i e_\ell.
    \end{split}
\end{equation*}
The variances of the first and the third terms are
\begin{equation*}
        \sum_{\ell = 1} ^n \sigma_\ell ^2 \left(\sum_{i = 1} ^n M_{W, i \ell} B_{ii} M_{W, ii}^{-1} x_i' \beta \right)^2
        \leq \max_i \sigma_i ^2 \sum_{i = 1} ^n B_{ii}^2 M_{W, ii}^{-2} (x_i' \beta)^2 \leq \max_i \sigma_i ^2 \max_i (x_i' \beta)^2 M_{W, ii}^{-2} \sum_{i = 1} ^n B_{ii}^2,
\end{equation*}
\begin{equation*}
        \sum_{\ell = 1} ^n \sigma_\ell ^2 \left(\sum_{i = 1} ^n M_{W, i \ell} B_{ii} M_{W, ii}^{-1} f(z_i) \right)^2
        \leq \max_i \sigma_i ^2 \sum_{i = 1} ^n B_{ii}^2 M_{W, ii}^{-2} f(z_i)^2 \leq \max_i \sigma_i ^2 \max_i f(z_i)^2 M_{W, ii}^{-2} \sum_{i = 1} ^n B_{ii}^2,
\end{equation*}
and of the second and the fourth
\begin{equation*}
    \sum_{i = 1} ^n B_{ii} ^2 \var[e_i ^2] \leq \max_i \E[e_i ^4] \sum_{i = 1} ^n B_{ii}^2,
\end{equation*}
\begin{equation*}
    \sum_{i = 1} ^n \sum_{\ell \neq i} (B_{ii}^2 M_{W, ii}^{-2} + B_{ii}M_{W, ii}^{-1} B_{\ell \ell} M_{W, \ell \ell}^{-1}) M_{W, i \ell} ^2 \sigma_i ^2 \sigma_\ell ^2 \leq 2 \max_i \sigma_i ^4 M_{W, ii}^{-2} \sum_{i = 1} ^n B_{ii} ^2.
\end{equation*}
Because each variance is bounded by $C \sum_{i = 1} ^n B_{ii} ^2$, to show that it is of smaller order than the variance of $\hat{\theta}$, we need $\var[\hat{\theta}]^{-1} \sum_{i = 1} ^n B_{ii} ^2 = o(1)$. It holds because
\begin{equation*}
    \var[\hat{\theta}]^{-1} \sum_{i = 1} ^n B_{ii} ^2 \leq \max_i v_i'v_i \var[\hat{\theta}]^{-1} \sum_{\ell = 1} ^r \lambda_\ell ^2 \leq \max_i v_i' v_i \max_i \sigma_i ^{-4} = o(1).
\end{equation*}
\\
\textit{Variance estimator consistency.} Now we show that the variance estimator is consistent, i.e. $\var[\hat{b}]^{-1} \widehat{\var}[\hat{b}] \overset{p}{\rightarrow} I_r$. For it we need to show that
\begin{equation*}
    \var[\vartheta'\hat{b}]^{-1} \left(\widehat{\var}[\vartheta'\hat{b}] - \var[\vartheta'\hat{b}]\right) = o_p(1), \quad \vartheta \in \mathbb{R}^r, \quad \vartheta'\vartheta = 1
\end{equation*}
for some non-random $\vartheta$. Rewrite the expression above as
\begin{equation}\label{delta_theta}
    \delta(\vartheta) := \sum_{i = 1} ^n v_i(\vartheta) (\hat{\sigma}_i ^2 - \sigma_i ^2),
\end{equation}
where
\begin{equation*}
    v_i(\vartheta) := \frac{(\vartheta' v_i)^2}{\sum_{i = 1} ^n \sigma_i ^2 (\vartheta' v_i)^2}.
\end{equation*}
We know that $\E[\delta(\vartheta)] = o_p(1)$ because by the triangle inequality, the Cauchy-Schwarz inequality, and $|\E[\hat{\sigma}_i ^2 - \sigma_i ^2]| = o_p(1)$,
\begin{equation*}
    | v(\vartheta) \E[\hat{\sigma}_i ^2 - \sigma_i ^2]| \leq (v(\vartheta)^2)^{1/2} \cdot \left(|\E[\hat{\sigma}_i ^2 - \sigma_i ^2]|^2\right)^{1/2} = o_p(1)
\end{equation*}
for $i = 1, \ldots, n$. The variance of $\delta(\vartheta)$ is 
\begin{equation*}
    \begin{split}
        \sum_{i = 1} ^n \delta(\vartheta)^2 \var[\hat{\sigma}_i ^2] &\leq \sum_{i = 1} ^n v_i(\vartheta)^4 \\
        &\leq \max_i \sigma_i ^{-4} \max_i v_i v_i' \frac{\vartheta'\vartheta}{\sum_{i = 1} ^n v_i v_i' \vartheta'\vartheta} \\
        &= \max_i \sigma_i ^{-4} \max_i v_i v_i' = o(1)
    \end{split}
\end{equation*}
because $\max_i v_i v_i' = o(1)$ by assumption.
\\
\textit{Asymptotic normality.} Our objective is to prove that
\begin{equation*}
    \var[\vartheta' \hat{b}]^{-1/2} \left(\vartheta'\hat{b} - \vartheta'b \right) \overset{d}{\rightarrow} \mathcal{N}(0, 1),
\end{equation*}
where $\hat{b} := Q' S_{xx}^{1/2} \hat{\beta}$, and $b := Q' S_{xx}^{1/2} \beta$. Lyapunov's condition implies that it is sufficient to show that
\begin{equation*}
    \var[\vartheta'\hat{b}]^{-2} \sum_{i = 1} ^n \E\left[\left(\vartheta' (\hat{b} - b)\right)^4\right] = o_p(1).
\end{equation*}
Because we have that
\begin{equation*}
    \begin{split}
        \vartheta'(\hat{b} - b) &= \vartheta'(Q'S_{xx}^{1/2} \hat{\beta} - Q'S_{xx}^{1/2} \beta) \\
        &= \sum_{i = 1} ^n \vartheta' v_i \left(f(z_i) + e_i\right),
    \end{split}
\end{equation*}
and $(a + b)^4 \leq C(a^4 + b^4)$ for some constant $C$, Lyapunov's condition is equivalent to
\begin{equation*}
    \var[\vartheta'\hat{b}]^{-2} \sum_{i = 1} ^n C \left(f(z_i)^4 + \E[e_i^4]\right) \cdot (\vartheta' v_i)^4 = o_p(1).
\end{equation*}
It holds because $\max_i |f(z_i)| = \mathcal{O}(1)$, and $\max_i \E[e_i ^4] = \mathcal{O}(1)$ by assumption, so that $\max_i f(z_i)^4 + \E[e_i^4] = \mathcal{O}(1)$, also $\max_i (\vartheta' v_i)^2 \leq \max_i v_i' v_i = o(1)$, $\sum_{i = 1} ^n (\vartheta' v_i)^2 = \sum_{i = 1} ^n \vartheta' v_i v_i' \allowbreak \vartheta = 1$, and $\var[\vartheta' \hat{b}]^{-2} \leq \max_i \sigma_i ^{-2} = \mathcal{O}(1)$. \hfill \qedsymbol

\paragraph*{Proof of Theorem \ref{big_r_norm}}

We derive the limiting distribution of $\hat{\theta}$ with growing rank based on the following result regarding the joint normality of independent and not necessarily identical random variables as in \textcite{kline}. 

Let $\{q_{n, i}\}_{i,n}$ be a triangular array of row-wise independent random variables with $\E[q_{n,i}] = 0$ and $\var[q_{n,i}] = \sigma^2 _{n, i}$, let $\{\dot{w}_{n,i}\}_{i,n}$ be a triangular array of non-random weights that satisfy $\sum_{i = 1} ^n \dot{w}_{n,i} \sigma_{n,i}^2 = 1$ for $\forall n$, and let $(Q_n)_n$ be a sequence of symmetric non-random matrices in $\mathbb{R}^{n \times n}$ with zeros on the diagonal and having $2 \sum_{i = 1} ^n \sum_{\ell \neq i} Q_{n, i\ell} ^2 \sigma_{n,i}^2 \sigma_{n, \ell}^2 = 1$. Define
\begin{equation*}
    \mathcal{S}_n := \sum_{i = 1} ^n \dot{w}_{n,i} q_{n,i}, \quad \mathcal{U}_n := \sum_{i = 1} ^n \sum_{\ell \neq i} Q_{n, i\ell} q_{n, i} q_{n, \ell}.
\end{equation*}
\begin{lem}\label{solv}
    If $\max_i \E[q_{n, i}^4] + \sigma_{n,i}^{-2} = O(1)$, (i) $\max_i \dot{w}_{n,i}^2 = o(1)$, and (ii) $\text{trace} (Q_{n}^4) = o(1)$, then $(\mathcal{S}_n, \mathcal{U}_n)' \overset{d}{\rightarrow} \mathcal{N}(0, I_2)$.
\end{lem}
\begin{proof}
    See Appendix B in \textcite{kline} and Appendix A2 in \textcite{solvsten}.
\end{proof}

The $U$-statistic representation holds because 
\begin{equation*}
    \begin{split}
        \hat{\theta} &= \sum_{i = 1} ^n y_i \tilde{w}_i' \hat{\gamma}_{-i} = \sum_{i = 1} ^n y_i \tilde{w}_i' \left(S_{ww}^{-1} - w_i w_i'\right) \sum_{\ell \neq i} w_\ell y_\ell \\
        &= \sum_{i = 1} ^n y_i \tilde{w}_i' \left(S_{ww}^{-1} + \frac{S_{ww}^{-1} w_i w_i' S_{ww}^{-1}}{1 - w_i' S_{ww}^{-1} w_i}\right) \sum_{\ell \neq i} w_\ell y_\ell \\
        &= \sum_{i = 1} ^n y_i \tilde{w}_i' \left( S_{ww}^{-1} + M_{W, ii}^{-1} S_{ww}^{-1} w_i w_i' S_{ww}^{-1} \right) \sum_{\ell \neq i} w_\ell y_\ell \\
        &= \sum_{i = 1} ^n y_i w_i' S_{ww}^{-1} \check{A} S_{ww}^{-1} \sum_{\ell \neq i} w_\ell y_\ell + \sum_{i = 1} ^n y_i w_i' S_{ww}^{-1} \check{A} M_{W, ii}^{-1} S_{ww}^{-1} w_i w_i' S_{ww}^{-1} \sum_{\ell \neq i} w_\ell y_\ell \\
        &= \sum_{i = 1} ^n \sum_{\ell \neq i} y_i y_\ell w_i' S_{ww}^{-1} \check{A} S_{ww}^{-1} w_\ell + \sum_{i = 1} ^n \sum_{\ell \neq i} y_i y_\ell M_{W, ii}^{-1} w_i' S_{ww}^{-1} \check{A} S_{ww}^{-1} w_i w_i' S_{ww}^{-1} w_\ell \\
        &= \sum_{i = 1} ^n \sum_{\ell \neq i} y_i y_\ell B_{W, i\ell} + y_i y_\ell M_{W, ii}^{-1} B_{W, ii} (1 - M_{W, i\ell}) = \sum_{i = 1} ^n \sum_{\ell \neq i} C_{i \ell} y_i y_\ell.
    \end{split}
\end{equation*}
Now, define $\check{e}_i := f(z_i) - p_k (z_i)' \alpha$ to be an approximation error (it is implicitly indexed by the unknown function but we omit this dependence for brevity). We can write the difference
\begin{align*}
        \hat{\theta} - \theta &= \sum_{i = 1} ^n \sum_{\ell \neq i} C_{i\ell} y_i y_\ell - \sum_{i = 1} ^n \gamma' w_i \tilde{w}_i' \gamma \\
        &= \sum_{i = 1} ^n \sum_{\ell \neq i} C_{i\ell} \left(x_i'\beta + f(z_i) + e_i \right) \left(x_\ell' \beta + f(z_\ell) + e_\ell \right) - \sum_{i = 1} ^n \gamma' w_i \tilde{w}_i' \gamma \\
        &= \sum_{i = 1} ^n \sum_{\ell \neq i} C_{i \ell} \left( \gamma' w_i + e_i + \check{e}_i \right) \left(\gamma' w_\ell + e_\ell + \check{e}_\ell\right) - \sum_{i = 1} ^n \gamma' w_i \tilde{w}_i' \gamma \\
        &= \sum_{i = 1} ^n \sum_{\ell \neq i} C_{i \ell} \gamma' w_i \gamma' w_\ell + \sum_{i = 1} ^n \sum_{\ell \neq i} C_{i \ell} (\gamma' w_i e_\ell + \gamma' w_\ell e_i) - \sum_{i = 1} ^n \gamma' w_i \tilde{w}_i' \gamma 
        \\&\qquad+ \sum_{i = 1} ^n \sum_{\ell \neq i} C_{i \ell} (\gamma' w_i \check{e}_\ell + \gamma' w_\ell \check{e}_i + e_i \check{e}_\ell + e_\ell \check{e}_i) + \sum_{i = 1} ^n \sum_{\ell \neq i} C_{i \ell} e_i e_\ell + \sum_{i = 1} ^n \sum_{\ell \neq i} C_{i \ell} \check{e}_i \check{e}_\ell\\
        &= \sum_{i = 1} ^n e_i \sum_{\ell \neq i} (\gamma' w_i + \gamma' w_\ell) C_{i\ell} + \sum_{i = 1} ^n \check{e}_i \sum_{\ell \neq i} (\gamma' w_i + \gamma' w_\ell + e_i + e_\ell)\\&\qquad+ \sum_{i = 1} ^n \sum_{\ell \neq i} C_{i \ell} e_i e_\ell + \sum_{i = 1} ^n \sum_{\ell \neq i} C_{i \ell} \check{e}_i \check{e}_\ell \\
        &= \sum_{i = 1} ^n (2 \tilde{w}_i' \gamma - \check{w}_i' \gamma) e_i + \sum_{i = 1} ^n \sum_{\ell \neq i} C_{i \ell} e_i e_\ell + \mathcal{O}_p (n k^{- \alpha_f}),
\end{align*}
where the last equality follows from defining $\check{w}_i := \sum_{\ell = 1} ^n M_{W, i\ell} \frac{B_{W, \ell \ell}}{1 - P_{W,\ell \ell}} w_\ell$, and the fact that by the Assumption \ref{asum1}
\begin{equation*}
    \check{e}_i ^2 \leq \left( f(z_i) - p_k (z_i)' \alpha \right)^2 = \E[\left( f(z_i) - p_k (z_i)' \alpha \right)^2] \leq \min_{\alpha \in \mathbb{R}^k} \E[\left( f(z_i) - p_k (z_i)' \alpha \right)^2] \leq Ck^{-2 \alpha_f}. \\
\end{equation*}
As $k \rightarrow \infty$, we have that
\begin{equation*}
    \hat{\theta} - \theta = \sum_{i = 1} ^n \left(2 \tilde{w}_i' \gamma - \check{w}_i' \gamma\right) e_i + \sum_{i = 1} ^n \sum_{\ell \neq i} C_{i \ell} e_i e_\ell + o_p(1).
\end{equation*}

Having dispensed with asymptotically negligible contributions to $\hat{\theta}$, asymptotic variance is 
\begin{equation*}
    \var[\hat{\theta}] = \sum_{i = 1} ^n (2\tilde{w}_i' \gamma - \check{w}_i' \gamma)^2 \sigma_i ^2 + 2 \sum_{i = 1} ^n \sum_{\ell \neq i} C_{i \ell} \sigma_i ^2 \sigma_\ell ^2,
\end{equation*}
that is, a sum of two components given by
\begin{equation*}
    \mathcal{V}_s := \sum_{i = 1} ^n (2\tilde{w}_i' \gamma - \check{w}_i' \gamma)^2 \sigma_i ^2, \quad \mathcal{V}_u := 2 \sum_{i = 1} ^n \sum_{\ell \neq i} C_{i \ell} \sigma_i ^2 \sigma_\ell ^2.
\end{equation*}
The normalized difference is then given by
\begin{equation*}
    \var[\hat{\theta}]^{-1/2} (\hat{\theta} - \theta) = \omega_1 \mathcal{S}_n + \omega_2 \mathcal{U}_n,
\end{equation*}
with $\omega_1 := \var[\hat{\theta}]^{-1/2} \mathcal{V}_s^{1/2}$, $\omega_2 := \var[\hat{\theta}]^{-1/2} \mathcal{V}_u^{1/2}$, and
\begin{equation*}
    \mathcal{S}_n := \mathcal{V}_s^{-1/2} \sum_{i = 1} ^n (2 \tilde{w}_i' \gamma - \check{w}_i' \gamma) e_i, \quad \mathcal{U}_n := \mathcal{V}_u ^{-1/2} \sum_{i = 1} ^n \sum_{\ell \neq i} C_{i \ell} e_i e_\ell.
\end{equation*}

Consider the case where the limit of $\omega_1$ is nonzero. If it is not, then asymptotic normality of the difference $\var[\hat{\theta}]^{-1/2} (\hat{\theta} - \theta)$ is implied by asymptotic normality of $\mathcal{U}_n$. Using notation of Lemma \ref{solv}, we have that $\dot{w}_i = \mathcal{V}_s^{-1/2} (2 \tilde{w}_i' \gamma - \check{w}_i' \gamma)$, and $Q_{i\ell} = \mathcal{V}_u ^{-1/2} C_{i\ell}$. 

To verify the condition \textit{(i)} of Lemma \ref{solv}, note that
\begin{equation*}
\begin{split}
    \max_i \dot{w}_i ^2 = \max_i \mathcal{V}_s^{-1} (2 \tilde{w}_i' \gamma - \check{w}_i' \gamma)^2 &\leq 
    \max_i 4 \mathcal{V}_s^{-1} \left((\tilde{w}_i'\gamma)^2 + (\check{w}_i'\gamma)^2\right) \\
    &= \max_i 4 \omega_1 ^{-2} \frac{(\tilde{w}_i'\gamma)^2 + (\check{w}_i'\gamma)^2}{\var[\hat{\theta}]} = o(1), 
\end{split}
\end{equation*}
where the last equality follows from Theorem \ref{big_r_norm} \textit{(i)}, and the nonzero limit of $\omega_1$.

For the condition \textit{(ii)} of Lemma \ref{solv}, denote $\tilde{A}_W := S_{ww}^{-1/2} \check{A} S_{ww}^{-1/2}$, and note that the first $r$ eigenvalues of the $\tilde{A}_W$ matrix are equal to eigenvalues of the $\tilde{A}$ matrix. Now, with constants $c_U$ and $c_L$ not dependent on $n$, we have that $\text{trace}(C^4) \leq c_U \cdot \text{trace}(B_W ^4) = c_U \cdot \text{trace}(\tilde{A}_W ^4) \leq c_U \lambda_1 ^2 \cdot \text{trace}(\tilde{A}_W ^2)$, and $\mathcal{V}_u \geq c_L \min_i \sigma_i ^4 \cdot \text{trace}(\tilde{A}_W) $, which implies
\begin{equation*}
    \text{trace} (Q^4) \leq \frac{c_U \lambda_1 ^2 \cdot \text{trace}(\tilde{A}_W ^2)}{\left(c_L \min_i \sigma_i ^4 \cdot \text{trace}(\tilde{A}_W ^2)\right)^2} = \mathcal{O}\left(\frac{\lambda_1 ^2}{\text{trace}(\tilde{A}_W ^2)}\right) = o(1),
\end{equation*}
where the last equality follows from Theorem \ref{big_r_norm} \textit{(ii)}. \hfill \qedsymbol

\section{Additional Tables}\label{a:tables}
Table \ref{tab:simulation_coverage} reports the finite-sample coverage diagnostic for the degree-5 leave-out estimator. Tables \ref{tab:descriptives}--\ref{tab:desc_cc} document the data and sample restrictions behind the estimation sample. Table \ref{tab:firm_controls_variation} documents the firm controls that drive the richer specification, and Tables \ref{tab:empirical_results_delta}--\ref{tab:empirical_diagnostics} report the specification ladder, its sensitivity to the nonlinear basis, the validation exercises, and numerical diagnostics.

\begin{table}[t]
\centering
\caption{Finite-sample coverage diagnostic for the degree-5 leave-out estimator}
\label{tab:simulation_coverage}
\vspace*{1mm}
\input{tables/simulation_coverage}
\note[Notes]{%
Entries report the share of 5,000 Monte Carlo replications in which the nominal 95\% confidence interval for the degree-5 leave-out estimator contains the quadratic-form target. The intervals use the estimated noncentral variance of the quadratic form and leave-one-out residual variances $\hat{\sigma}_i^2$. \emph{NA} means that the complete degree-5 basis is not estimable. This table is a finite-sample diagnostic, not an applied standard-error claim.}
\end{table}

\begin{table}[h!]
    \caption{Descriptive statistics of the cleaned worker--firm panel}
    \label{tab:descriptives}
    \centering
        \input{tables/descriptive_fdp}
\note[Notes]{%
The table uses the cleaned full-time worker--firm panel for workers aged 20--65 in 2008--2018, before the non-missing-control, connected-component, and leave-match-out restrictions. Panels A and C report worker-year and firm-year statistics, respectively: means, standard deviations, and the 25th, 50th, and 75th percentiles. Units are shown in the row labels. Panel B reports counts and sample shares; its percentages use worker-year rows as the denominator, the job-to-job move share uses non-missing switch indicators, and categorical missing values are shown explicitly. Panel D reports counts.}
\end{table}

\begin{table}[h!]
    \caption{Missingness of estimation variables}
    \label{tab:descriptive_coverage}
    \centering
        \input{tables/descriptive_coverage}
\note[Notes]{%
Entries report the share of worker-year observations in the cleaned analysis sample (2008--2018) with missing values for each variable required by at least one empirical specification, including the outcome, identifiers, year, and controls. This table uses worker-year denominators and is computed before the non-missing-control, connected-component, and leave-match-out restrictions. The firm-control diagnostic in Table \ref{tab:firm_controls_variation} uses firm-year denominators and an imputed employment field, so its missingness rates are not directly comparable. The main empirical tables use the Panel B leave-match-out sample described in Table \ref{tab:desc_cc}.}
\end{table}

\begin{table}[t]
    \caption{Sample flow into the leave-match-out estimation set}
    \label{tab:desc_cc}
    \centering
        \input{tables/descriptive_cc_master}
\note[Notes]{%
The table reports sample flow from loaded QP rows through the cleaned worker-year panel, non-missing controls, the largest worker--firm connected component, and the final leave-match-out set used by the full-regressor cluster-fold point correction. $N$ is the number of observations at each stage; Workers and Firms are unique identifiers, Movers are workers observed at more than one firm, and Matches are unique worker--firm pairs.
The Panel A rows describe the larger parsimonious candidate sample. The main empirical tables use the Panel B leave-match-out set for both panels.}
\end{table}

\begin{table}[t]
    \caption{Firm controls: within-firm time variation and missingness}
    \label{tab:firm_controls_variation}
    \centering
        \input{tables/firm_controls_variation}
\note[Notes]{%
Panel A uses the cleaned 2008--2018 firm-year panel used to form the estimation sample, not the final leave-match-out sample. It reports within-firm time variation for the Panel B firm controls, restricting to firms observed for at least two years. Each row uses its own non-missing firm-years, so row counts differ across controls and need not equal the joint estimation sample.
The within-firm share is the share of total variance on the estimation scale attributable to within-firm changes: log employment and $\log(1+x)$ for fixed assets and intermediate inputs. Zero within SD is the share of firms with at least two observed years and zero within-firm standard deviation of the transformed control. Panel B reports year-by-year missingness rates in the firm-year panel. The worker-count column has zero missingness because employment is imputed from the matched employer--employee data; fixed assets are unavailable in 2008--2009 and are about 18--24\% missing thereafter; intermediate inputs are about 18--26\% missing throughout the sample. The final column reports the share of firm-years with workers, fixed assets, and intermediate inputs all observed. These firm-year rates are not directly comparable to the worker-year rates in Table \ref{tab:descriptive_coverage}.}
\end{table}

\begin{table}[t]
\centering
\caption{Stepwise changes in wage variance decomposition}
\label{tab:empirical_results_delta}
\vspace*{1mm}
\input{tables/empirical_results_delta}
\note[Notes]{%
Entries are next-specification minus comparison-specification differences in bias-corrected variance-component shares from the full-regressor cluster-fold point correction on the common Panel B leave-match-out estimation sample. The five rows are, respectively: linear additive controls minus the AKM baseline; heterogeneous linear controls minus linear additive controls; nonlinear homogeneous controls minus linear additive controls; the full model minus nonlinear homogeneous controls; and the full model minus heterogeneous linear controls. A positive value denotes an increase.
Starting from the linear additive specification, one path adds heterogeneity at degree 1 and then nonlinearity with interactions; the other adds nonlinearity without interactions and then heterogeneity in the nonlinear basis. Both paths arrive at the full model. These are point-estimate differences; no statistical standard errors or formal tests are reported.}
\end{table}

\begin{table}[t]
\centering
\caption{Alternative series bases and the wage variance decomposition}
\label{tab:empirical_basis_robustness}
\vspace*{1mm}
\input{tables/empirical_basis_robustness}
\note[Notes]{%
Entries report full-regressor cluster-fold bias-corrected variance-component shares and the residualized-wage variance ratio on the common Panel B leave-match-out estimation sample. \emph{Common} denotes the nonlinear homogeneous specification; \emph{Group-specific} denotes the full specification that interacts the basis with worker groups. The residualized-wage ratio uses $\tilde y$ as defined in the note to Table \ref{tab:empirical_results_residualized}.
The polynomial and Hermite specifications use complete degree-5 multivariate bases. The B-spline specification is additive across continuous inputs and uses cubic bases with quantile-spaced knots; its spline counts match the polynomial model's realized control dimension.
All specifications use the same 20 outcome-independent folds, treating mover matches as clusters and stayer person-years as singleton clusters. The randomized correction satisfies a numerical-precision rule: the Monte Carlo standard error of each correction term is at most $6 \times 10^{-4}$ of total log-wage variance. This is a numerical-precision check, not a statistical inference statement. The table reports point estimates only.}
\end{table}

\begin{table}[t]
\centering
\caption{Validation of the flexible covariate specification}
\label{tab:empirical_covariate_validation}
\vspace*{1mm}
\input{tables/empirical_covariate_diagnostics}
\note[Notes]{%
These diagnostics use the common Panel B leave-match-out sample. Panel A compares the nonlinear specifications with their nested linear benchmarks. Partial $R^2$ is one minus the ratio of the rich-model mean squared error to the restricted-model mean squared error. The held-out result uses 94,493 observations from one held-out fold of an outcome-independent, graph-aware 100-fold assignment; the assignment preserves regressor support and connectedness in the training sample.
Panel B injects a nonlinear component that is orthogonal to the heterogeneous-linear design and has variance equal to 5\% of that model's residual variance. The recovery statistics compare the rich model's incremental fitted component with the injected signal.
Panel C compares the plug-in worker and firm coefficients from the heterogeneous-linear and full joint regressions. The root mean squared (RMS) changes and decile movements use observation weights. These diagnostics are descriptive and are not the bias-corrected variance components reported in the main tables; no statistical standard errors are reported.}
\end{table}

\begin{table}[t]
\centering
        \caption{Control dimension and specification diagnostics}
    \label{tab:empirical_diagnostics}
    \centering
        \input{tables/empirical_results_diagnostics}
\note[Notes]{%
$k$ counts the observed-control and basis columns in the joint augmented regression; worker, firm, and year fixed-effect columns are excluded. Both panels use the common Panel B leave-match-out sample; Panel A changes the controls only. This table reports control dimension only. Sample-flow counts are reported separately in Table \ref{tab:desc_cc}.}
\end{table}

\end{appendix}

\end{document}

%% file: tables/simulation_results.tex
\begin{tabular*}{\textwidth}[]{p{4.5cm}@{\extracolsep\fill}ccccc}
\toprule
Scenario & PI(1) & LOO(1) & LOO(3) & LOO(5) & RMSE(5) \\
\midrule
\multicolumn{6}{l}{\textit{Approximation and nonlinearity}}\\
\midrule
Mild nonlinearity & 0.01 & 0.00 & 0.00 & 0.00 & 0.05 \\
Linear benchmark & 0.00 & 0.00 & 0.00 & 0.00 & 0.05 \\
Strong nonlinearity & $-$0.04 & $-$0.04 & $-$0.02 & $-$0.01 & 0.05 \\
\midrule
\multicolumn{6}{l}{\textit{Strong radial design: one-feature changes}}\\
\midrule
Heteroskedasticity & $-$0.03 & $-$0.04 & $-$0.02 & $-$0.01 & 0.05 \\
Heavy-tailed errors & $-$0.03 & $-$0.04 & $-$0.01 & $-$0.01 & 0.05 \\
Higher nuisance dimension & $-$0.05 & $-$0.05 & $-$0.02 & NA & NA \\
Many regressors & $-$0.01 & $-$0.03 & $-$0.02 & $-$0.01 & 0.06 \\
High leverage & $-$0.02 & $-$0.03 & $-$0.01 & 0.00 & 0.02 \\
Larger sample & 0.02 & 0.01 & $-$0.01 & 0.00 & 0.03 \\
\bottomrule
\end{tabular*}

%% file: tables/data/gakm_akm_calibrated/simulation_results_akm_calibrated.tex
\begin{tabular*}{\textwidth}[]{p{2.25cm}p{1.65cm}@{\extracolsep\fill}rrr}
\toprule
DGP & Nuisance fit & KSS bias (MC SD) & HO bias & PI bias \\
\midrule
Linear DGP & Flexible & 0.06\% (0.50\%) & 1.24\% & 1.90\% \\
Linear DGP & Linear & 0.06\% (0.48\%) & 1.25\% & 1.90\% \\
Nonlinear DGP & Flexible & 0.07\% (0.62\%) & 1.31\% & 1.98\% \\
Nonlinear DGP & Linear & 0.18\% (0.62\%) & 1.60\% & 2.37\% \\
\bottomrule
\end{tabular*}

%% file: tables/empirical_results.tex
\begin{tabular*}{\textwidth}[]{p{4.8cm}@{\extracolsep\fill}ccc}
\toprule
Model & $\sigma_\alpha^2$ & $\sigma_\psi^2$ & $\cov(\alpha,\psi)$ \\
\midrule
\multicolumn{4}{l}{\textit{Panel A: Parsimonious controls (Panel B sample)}}\\
\midrule
AKM baseline & 0.5507 & 0.1435 & 0.0803 \\
Linear additive controls & 0.5326 & 0.1367 & 0.0738 \\
Heterogeneous linear controls & 0.5388 & 0.1366 & 0.0748 \\
Nonlinear homogeneous controls & 0.5483 & 0.1358 & 0.0775 \\
Full model & 0.5619 & 0.1371 & 0.0767 \\
\midrule
\multicolumn{4}{l}{\textit{Panel B: Worker and firm-input controls}}\\
\midrule
AKM baseline & 0.5507 & 0.1435 & 0.0803 \\
Linear additive controls & 0.4736 & 0.1214 & 0.0465 \\
Heterogeneous linear controls & 0.4506 & 0.1198 & 0.0446 \\
Nonlinear homogeneous controls & 0.4913 & 0.1199 & 0.0410 \\
Full model & 0.4904 & 0.1166 & 0.0420 \\
\bottomrule
\end{tabular*}

%% file: tables/empirical_results_residualized_variance.tex
\begin{tabular*}{0.70\textwidth}[]{p{4.8cm}@{\extracolsep\fill}c}
\toprule
Model & $\var(\tilde y)/\var(y)$ \\
\midrule
\multicolumn{2}{l}{\textit{Panel A: Parsimonious controls (Panel B sample)}}\\
\midrule
AKM baseline & 0.9985 \\
Linear additive controls & 0.9561 \\
Heterogeneous linear controls & 0.9642 \\
Nonlinear homogeneous controls & 0.9783 \\
Full model & 0.9913 \\
\midrule
\multicolumn{2}{l}{\textit{Panel B: Worker and firm-input controls}}\\
\midrule
AKM baseline & 0.9985 \\
Linear additive controls & 0.8271 \\
Heterogeneous linear controls & 0.7982 \\
Nonlinear homogeneous controls & 0.8319 \\
Full model & 0.8271 \\
\bottomrule
\end{tabular*}

%% file: tables/simulation_coverage.tex
\begin{tabular*}{\textwidth}[]{p{4.5cm}@{\extracolsep\fill}c}
\toprule
Scenario & 95\% coverage (LOO(5)) \\
\midrule
\multicolumn{2}{l}{\textit{Approximation and nonlinearity}}\\
\midrule
Mild nonlinearity & 0.93 \\
Linear benchmark & 0.93 \\
Strong nonlinearity & 0.90 \\
\midrule
\multicolumn{2}{l}{\textit{Strong radial design: one-feature changes}}\\
\midrule
Heteroskedasticity & 0.90 \\
Heavy-tailed errors & 0.91 \\
Higher nuisance dimension & NA \\
Many regressors & 0.97 \\
High leverage & 0.86 \\
Larger sample & 0.95 \\
\bottomrule
\end{tabular*}

%% file: tables/descriptive_fdp.tex
\begin{tabular*}{\textwidth}[]{p{5.5cm}@{\extracolsep\fill}ccccc}
\toprule
& Mean & SD & p25 & Median & p75 \\
\midrule
\multicolumn{6}{l}{\textit{A. Worker-year panel}}\\
Age & 39.20 & 10.52 & 31.00 & 38.00 & 47.00 \\
Months worked & 8.19 & 1.69 & 8.00 & 9.00 & 9.00 \\
Tenure (months) & 103.57 & 103.76 & 22.00 & 69.00 & 153.00 \\
Annual wage (k€) & 13.92 & 10.46 & 7.99 & 10.53 & 16.02 \\
Hourly wage (€) & 7.17 & 5.17 & 4.15 & 5.33 & 7.99 \\
Number of jobs & 1.01 & 0.12 & 1.00 & 1.00 & 1.00 \\
\midrule
\multicolumn{6}{l}{\textit{B. Worker composition and mobility}}\\
Job-to-job moves (count) & 1,803,534 & -- & -- & -- & -- \\
Job-to-job movers (\%) & 9.85\% & -- & -- & -- & -- \\
Female workers (\%) & 42.14\% & -- & -- & -- & -- \\
Education: At most primary (\%) & 47.91\% & -- & -- & -- & -- \\
Education: Secondary (\%) & 34.25\% & -- & -- & -- & -- \\
Education: At least bachelor's (\%) & 17.52\% & -- & -- & -- & -- \\
Education missing (\%) & 0.32\% & -- & -- & -- & -- \\
Qualification: Specialized workers (\%) & 60.44\% & -- & -- & -- & -- \\
Qualification: Generic workers (\%) & 15.27\% & -- & -- & -- & -- \\
Qualification missing (\%) & 11.88\% & -- & -- & -- & -- \\
Qualification: Top managers (\%) & 7.10\% & -- & -- & -- & -- \\
Qualification: Middle managers (\%) & 5.31\% & -- & -- & -- & -- \\
Lisboa (\%) & 35.13\% & -- & -- & -- & -- \\
Norte (\%) & 35.06\% & -- & -- & -- & -- \\
Centro (\%) & 18.28\% & -- & -- & -- & -- \\
Alentejo (\%) & 4.54\% & -- & -- & -- & -- \\
Algarve (\%) & 3.67\% & -- & -- & -- & -- \\
Region (NUT2) other (\%) & 3.32\% & -- & -- & -- & -- \\
\midrule
\multicolumn{6}{l}{\textit{C. Firm-year panel}}\\
Workers & 13.35 & 119.97 & 3.00 & 4.00 & 9.00 \\
Fixed assets (M€) & 0.72 & 26.38 & 0.01 & 0.03 & 0.14 \\
Intermediate inputs (M€) & 0.72 & 22.49 & 0.03 & 0.06 & 0.19 \\
\midrule
\multicolumn{6}{l}{\textit{D. Panel counts}}\\
Unique workers & 3,532,497 & -- & -- & -- & -- \\
Unique firms & 484,704 & -- & -- & -- & -- \\
Worker-year observations & 18,314,740 & -- & -- & -- & -- \\
\bottomrule
\end{tabular*}

%% file: tables/descriptive_coverage.tex
\begin{tabular*}{0.68\textwidth}[]{p{5.5cm}@{\extracolsep\fill}r}
\toprule
Variable & Missing share \\
\midrule
Age & 0.00\% \\
Education & 0.32\% \\
Firm ID & 0.00\% \\
Fixed assets & 25.50\% \\
Intermediate inputs & 8.26\% \\
Log hourly wage & 0.01\% \\
Qualification & 11.88\% \\
Gender & 0.00\% \\
Worker ID & 0.00\% \\
Workers & 8.26\% \\
Year & 0.00\% \\
\bottomrule
\end{tabular*}

%% file: tables/descriptive_cc_master.tex
\begingroup
\small
\begin{tabular*}{\textwidth}[]{p{4.2cm}@{\extracolsep\fill}rrrrr}
\toprule
Stage & $N$ & Workers & Firms & Movers & Matches \\
\midrule
\multicolumn{6}{l}{\textit{Panel A: Parsimonious controls}}\\
\midrule
Loaded QP rows & 23,304,646 & 4,217,716 & 533,360 & 1,758,967 & 7,104,038 \\
Cleaned analysis sample & 18,314,740 & 3,532,497 & 484,704 & 1,231,386 & 5,291,365 \\
Non-missing controls & 18,312,694 & 3,532,350 & 484,696 & 1,231,231 & 5,290,982 \\
Largest connected component & 17,256,654 & 3,282,116 & 333,180 & 1,212,893 & 5,021,308 \\
Leave-match-out set & 15,232,458 & 2,432,843 & 193,487 & 1,075,443 & 3,969,503 \\
\midrule
\multicolumn{6}{l}{\textit{Panel B: Worker and firm-input controls}}\\
\midrule
Loaded QP rows & 23,304,646 & 4,217,716 & 533,360 & 1,758,967 & 7,104,038 \\
Cleaned analysis sample & 18,314,740 & 3,532,497 & 484,704 & 1,231,386 & 5,291,365 \\
Non-missing controls & 11,892,426 & 2,642,775 & 296,549 & 663,784 & 3,487,931 \\
Largest connected component & 10,891,978 & 2,392,780 & 180,478 & 653,930 & 3,227,750 \\
Leave-match-out set & 9,439,305 & 1,719,643 & 105,442 & 578,409 & 2,457,963 \\
\bottomrule
\end{tabular*}
\endgroup

%% file: tables/firm_controls_variation.tex
\begingroup
\begin{tabular*}{\textwidth}[]{p{2.7cm}@{\extracolsep\fill}>{\raggedleft\arraybackslash}p{1.65cm}>{\raggedleft\arraybackslash}p{1.55cm}>{\raggedleft\arraybackslash}p{1.35cm}>{\raggedleft\arraybackslash}p{1.55cm}>{\raggedleft\arraybackslash}p{1.30cm}}
\toprule
& Non-missing & Firms ($\geq$2y) & Mean years & Within share & Zero within SD \\
& firm-year & & (firms $\geq$2y) & (\%) & (\%) \\
\midrule
\multicolumn{6}{l}{\textit{A. Firm-control time variation in estimation sample}}\\
Employment & 2,397,757 & 383,679 & 5.99 & 7.88\% & 27.11\% \\
Fixed assets & 1,504,836 & 253,373 & 5.71 & 17.75\% & 5.99\% \\
Intermediate inputs & 1,874,310 & 290,932 & 6.22 & 9.75\% & 0.72\% \\
\midrule
\multicolumn{6}{l}{\textit{B. Missingness by year (firm-year panel)}}\\
Year & Firm-years & Workers miss. & Fixed assets miss. & Intermed. inputs miss. & All controls obs. \\
2008 & 253,005 & 0.00\% & 100.00\% & 26.27\% & 0.00\% \\
2009 & 245,077 & 0.00\% & 100.00\% & 25.37\% & 0.00\% \\
2010 & 223,267 & 0.00\% & 23.66\% & 23.66\% & 76.34\% \\
2011 & 218,229 & 0.00\% & 23.01\% & 23.01\% & 76.99\% \\
2012 & 202,744 & 0.00\% & 22.47\% & 22.47\% & 77.53\% \\
2013 & 198,597 & 0.00\% & 21.20\% & 21.20\% & 78.80\% \\
2014 & 202,625 & 0.00\% & 20.36\% & 20.36\% & 79.64\% \\
2015 & 207,037 & 0.00\% & 19.89\% & 19.89\% & 80.11\% \\
2016 & 211,892 & 0.00\% & 19.39\% & 19.39\% & 80.61\% \\
2017 & 215,908 & 0.00\% & 18.80\% & 18.80\% & 81.20\% \\
2018 & 219,376 & 0.00\% & 18.24\% & 18.24\% & 81.76\% \\
\bottomrule
\end{tabular*}
\endgroup

%% file: tables/empirical_results_delta.tex
\begin{tabular*}{\textwidth}[]{p{6.2cm}@{\extracolsep\fill}ccc}
\toprule
Change & $\Delta\sigma_\alpha^2$ & $\Delta\sigma_\psi^2$ & $\Delta\cov(\alpha,\psi)$ \\
\midrule
\multicolumn{4}{l}{\textit{Panel A: Parsimonious controls (Panel B sample)}}\\
\midrule
Add linear observables & -0.0181 & -0.0068 & -0.0065 \\
Add heterogeneity at degree 1 & 0.0061 & -0.0001 & 0.0010 \\
Add nonlinearity without interactions & 0.0157 & -0.0009 & 0.0037 \\
Add heterogeneity in nonlinear basis & 0.0136 & 0.0013 & -0.0008 \\
Add nonlinearity with interactions & 0.0232 & 0.0005 & 0.0018 \\
\midrule
\multicolumn{4}{l}{\textit{Panel B: Worker and firm-input controls}}\\
\midrule
Add linear observables & -0.0771 & -0.0221 & -0.0338 \\
Add heterogeneity at degree 1 & -0.0230 & -0.0016 & -0.0019 \\
Add nonlinearity without interactions & 0.0177 & -0.0016 & -0.0056 \\
Add heterogeneity in nonlinear basis & -0.0009 & -0.0032 & 0.0010 \\
Add nonlinearity with interactions & 0.0398 & -0.0032 & -0.0026 \\
\bottomrule
\end{tabular*}

%% file: tables/empirical_basis_robustness.tex
\begin{tabular*}{\textwidth}[]{p{3.0cm}p{3.2cm}@{\extracolsep\fill}rrrr}
\toprule
Basis & Nonlinear specification & $\sigma_\alpha^2$ & $\sigma_\psi^2$ & $\cov(\alpha,\psi)$ & $\var(\tilde y)/\var(y)$ \\
\midrule
\multicolumn{6}{l}{\textit{Panel A: Parsimonious controls (Panel B sample)}}\\
\midrule
Degree-5 polynomial & Common & 0.5483 & 0.1358 & 0.0775 & 0.9783 \\
Degree-5 polynomial & Group-specific & 0.5619 & 0.1371 & 0.0767 & 0.9913 \\
Degree-5 Hermite & Common & 0.5418 & 0.1358 & 0.0770 & 0.9709 \\
Degree-5 Hermite & Group-specific & 0.5491 & 0.1371 & 0.0757 & 0.9765 \\
Cubic B-spline & Common & 0.5505 & 0.1358 & 0.0776 & 0.9808 \\
Cubic B-spline & Group-specific & 0.5502 & 0.1371 & 0.0751 & 0.9764 \\
\midrule
\multicolumn{6}{l}{\textit{Panel B: Worker and firm-input controls}}\\
\midrule
Degree-5 polynomial & Common & 0.4913 & 0.1199 & 0.0410 & 0.8319 \\
Degree-5 polynomial & Group-specific & 0.4904 & 0.1166 & 0.0420 & 0.8271 \\
Degree-5 Hermite & Common & 0.4900 & 0.1199 & 0.0409 & 0.8306 \\
Degree-5 Hermite & Group-specific & 0.4908 & 0.1165 & 0.0424 & 0.8286 \\
Cubic B-spline & Common & 0.4925 & 0.1159 & 0.0400 & 0.8276 \\
Cubic B-spline & Group-specific & 0.4740 & 0.1143 & 0.0411 & 0.8070 \\
\bottomrule
\end{tabular*}

%% file: tables/empirical_covariate_diagnostics.tex
\begin{tabular*}{\textwidth}{@{\extracolsep{\fill}}lrr}
\toprule
Diagnostic & Worker effect & Firm effect \\
\midrule
\multicolumn{3}{l}{\textit{Panel A. Incremental wage fit}} \\
In-sample partial $R^2$: homogeneous nonlinear vs. additive linear & \multicolumn{2}{c}{0.39\%} \\
In-sample partial $R^2$: interacted nonlinear vs. heterogeneous linear & \multicolumn{2}{c}{1.25\%} \\
Held-out partial $R^2$: interacted nonlinear vs. heterogeneous linear & \multicolumn{2}{c}{1.08\%} \\
Held-out observations & \multicolumn{2}{c}{94{,}493} \\
\addlinespace
\multicolumn{3}{l}{\textit{Panel B. Injected nonlinear signal}} \\
Signal variance relative to restricted residual variance & \multicolumn{2}{c}{5.0\%} \\
Recovery slope & \multicolumn{2}{c}{0.998} \\
Recovery correlation & \multicolumn{2}{c}{0.994} \\
\addlinespace
\multicolumn{3}{l}{\textit{Panel C. Fixed-effect reallocation, heterogeneous linear to full model}} \\
RMS coefficient change & 0.178 & 0.033 \\
Coefficient correlation & 0.956 & 0.994 \\
Observation weight attached to effects changing decile & 54.5\% & 24.0\% \\
\bottomrule
\end{tabular*}

%% file: tables/empirical_results_diagnostics.tex
\begin{tabular*}{0.56\textwidth}[]{p{5.0cm}@{\extracolsep\fill}c}
\toprule
Model & $k$ \\
\midrule
\multicolumn{2}{l}{\textit{Panel A: Parsimonious controls (Panel B sample)}}\\
\midrule
AKM baseline & 0 \\
Linear additive controls & 2 \\
Heterogeneous linear controls & 3 \\
Nonlinear homogeneous controls & 4 \\
Full model & 9 \\
\midrule
\multicolumn{2}{l}{\textit{Panel B: Worker and firm-input controls}}\\
\midrule
AKM baseline & 0 \\
Linear additive controls & 10 \\
Heterogeneous linear controls & 102 \\
Nonlinear homogeneous controls & 129 \\
Full model & 3004 \\
\bottomrule
\end{tabular*}